\documentclass[12pt]{article}
\usepackage{graphicx,psfrag,epsf, xcolor}
\usepackage{enumerate}
\usepackage{natbib}
\usepackage{url} 

\usepackage{amssymb, amsmath, mathtools}
\usepackage{amsthm}
\newtheorem{thm}{Theorem} 
\newtheorem{cor}{Corollary} 
 
\newtheorem{prop}{Proposition}

\usepackage{algorithm}
\usepackage{algorithmic}

\usepackage{color}
\usepackage{lipsum}

\usepackage{fixme}
\usepackage{siunitx,booktabs}

\DeclareMathOperator{\Exp}{E}
\DeclareMathOperator{\Hop}{H}
\DeclareMathOperator{\Gop}{G}

\newcommand{\G}{{\mathcal G}}

\newcommand{\NN}{\mathbb{N}}
\newcommand{\RR}{\mathbb{R}}
\DeclareMathOperator{\ve}{vec}

\DeclareMathOperator*{\argmin}{arg\,min}

\newcommand{\X}{{\mathcal{X}}}

\newcommand{\blind}{0}

\begin{document}

\def\spacingset#1{\renewcommand{\baselinestretch}%
{#1}\small\normalsize} \spacingset{1}


\if0\blind
{
  \title{\bf Penalized estimation in large-scale generalized linear array models}
  \author{Adam Lund\thanks{Part of the Dynamical Systems Interdisciplinary Network, University of Copenhagen.}\\
Department of Mathematical Sciences, 
University of Copenhagen,\\
   Martin Vincent\thanks{
Supported by The Danish Cancer Society and The Danish Strategic Research Council/Innovation Fund Denmark.}\hspace{.2cm} \\
   Department of Mathematical Sciences, 
University of Copenhagen\\
    and \\
Niels Richard Hansen$^\ast$ \\
   Department of Mathematical Sciences, 
University of Copenhagen
}
  \maketitle
} \fi

\if1\blind
{
  \bigskip
  \bigskip
  \bigskip
  \begin{center}
    {\LARGE\bf Title}
\end{center}
  \medskip
} \fi

\bigskip
\begin{abstract}
  Large-scale generalized linear array models (GLAMs) can be
  challenging to fit. Computation and storage of its tensor product
  design matrix can be impossible due to time and memory constraints,
  and previously considered design matrix free algorithms do not scale
  well with the dimension of the parameter vector. A new design matrix free algorithm is
  proposed for computing the penalized maximum likelihood estimate
  for GLAMs, which, in particular, handles
  nondifferentiable penalty functions. The proposed algorithm is
  implemented and available via the R package \verb+glamlasso+. It
  combines several ideas -- previously considered separately -- to
  obtain sparse estimates while at the same time efficiently
  exploiting the GLAM structure. In this paper the convergence of the
  algorithm is treated and the performance of its implementation
  is investigated and compared to that of \verb+glmnet+ on simulated
  as well as real data. It is shown that the computation time for
  \verb+glamlasso+ scales favorably with the size of the problem when
  compared to \verb+glmnet+.  Supplemental materials are available online.
  
\end{abstract}

\noindent%
\vfill

\newpage
\spacingset{1.} 

\section{Introduction}\label{sec:intro}
The generalized linear array models (GLAMs) were
introduced in \cite{currie2006} as generalized linear models (GLMs) where the
observations can be organized in an array and the design matrix has a
tensor product structure. One main application treated in
\cite{currie2006} --  that will also be central to this paper -- is multivariate
smoothing where data is observed on a multidimensional grid.


In this paper we present results on 3-dimensional smoothing for
  two quite different real data sets where the aim was to extract a
  smooth mean signal. The first data set contains voltage sensitive
  dye recordings of spiking neurons in a live ferret brain and was
  modeled in a Gaussian GLAM framework. The second data set contains
  all registered Medallion taxi pick ups in New York City during 2013
  and was modeled in a Poisson GLAM framework. In both examples we
  fitted an $\ell_1$-penalized B-spline basis expansion to obtain a
  clear signal. For the taxi data we also demonstrate how the
  $\ell_1$-penalized fit lead to a lower error, compared to the
  non-penalized fit, when trying to predict missing observations.
Other potential applications include factorial designs and contingency
tables.

 \cite{currie2006} showed how the structure of
GLAMs can be exploited for computing the maximum likelihood estimate
and other quantities of importance for statistical inference. The
penalized maximum likelihood estimate for a quadratic penalty
function can also be computed easily by similar methods.    
The computations are simple to implement efficiently in any high level
language like \verb+R+ or \verb+MATLAB+ that supports fast numerical
linear algebra routines. They exploit the GLAM structure to carry out
linear algebra operations involving only the tensor factors -- called array
arithmetic, see also \cite{deboor1979} and \cite{buis1996} -- and they
avoid forming the design matrix. This design matrix free approach offers benefits in terms of memory as well
as time usage compared to standard GLM computations. 

The approach in \cite{currie2006} has some limitations when the
dimension $p$ of the parameter vector becomes large. The $p \times p$ weighted 
cross-product of the design matrix has to be computed, and though this
computation can benefit from the GLAM structure, a linear equation in
the parameter vector remains to be solved. The computations can become prohibitive
for large $p$. Moreover, the approach does not readily generalize to
non-quadratic penalty functions like the $\ell_1$-penalty  or for that matter non-convex penalty functions like the smoothly clipped absolute deviation (SCAD)  penalty.

In this paper we investigate the computation of the penalized maximum
likelihood estimate in GLAMs for a general convex penalty
function. However, we note that by employing the multi-step
  adaptive lasso (MSA-lasso) algorithm from Sections 2.8.5 and 2.8.6
  in \cite{buhlmann2011} our algorithm can easily be extended to
  handle non-convex penalty functions. This modification is already
  implemented in the R-package \texttt{glamlasso} for the
  SCAD-penalty, see \cite{lund2016}. The convergence results presented
  in this paper are, however, only valid for a convex penalty.

 Algorithms considered in the literature hitherto for
$\ell_1$-penalized estimation in GLMs, see e.g. \cite{friedman2010}, cannot easily benefit from the GLAM structure, and typically they need the
design matrix explicitly or at least direct access to its columns. Our proposed algorithm based on proximal
operators is design matrix free -- in the sense that the tensor
product design matrix need not be computed -- and can exploit the GLAM structure,
which results in an algorithm that is both memory and time efficient. 

The paper is organized as follows. In Section \ref{sec:glam} GLAMs are introduced. In Section
\ref{sec:alg} our proposed GD-PG algorithm for computing the
penalized maximum likelihood estimate  is described. Section \ref{sec:multdimsmooth}
presents two multivariate smoothing examples where the algorithm is
used to fit GLAMs. This section includes a benchmark comparison between our
implementation of the GD-PG algorithm in the R package
\verb+glamlasso+ and the algorithm implemented in \verb+glmnet+. Section \ref{sec:conv} presents a
convergence analysis of the proposed algorithm. In Section
\ref{sec:glamops} a number of details on how the algorithm is
implemented in \verb+glamlasso+ are collected. This includes details
on how the GLAM structure is exploited, and the section also presents
further benchmark results. Section
\ref{sec:disc} concludes the paper with a discussion. Some technical
and auxiliary definitions and results are presented in two appendices.  

\section{Generalized linear array models}\label{sec:glam}
A generalized linear model (GLM) is a regression model of $n$
independent real valued random variables $\mathcal{Y}_1, \ldots,
\mathcal{Y}_n$, see \cite{nelder1972}. A generalized linear array
model (GLAM) is a GLM with some additional structure of the data. We first introduce GLMs and then the special data structure for GLAMs.

With $X$ an $n\times p$ design matrix, the linear predictor $\eta :
\mathbb{R}^p \to \mathbb{R}^n$ is defined as
\begin{alignat}{4}
\eta(\theta) \coloneqq X\theta
\end{alignat}
for $\theta \in \mathbb{R}^p$. With $g:\RR\to\RR$ denoting the link
function, the mean value of $\mathcal{Y}_i$ is given in terms of
$\eta_i(\theta)$ via the equation 
\begin{alignat}{4}\label{one}
g(\Exp(\mathcal{Y}_i)) = \eta_i(\theta).
\end{alignat}
The link function $g$ is throughout assumed invertible with a continuously differentiable inverse.

The distribution of $\mathcal{Y}_i$ is, furthermore, assumed to belong
to an exponential family, see Appendix \ref{sec:expfam}, which implies that the log-likelihood, $\theta\mapsto l(\eta(\theta))$, is given in terms of the linear predictor. With $y=(y_1,\ldots,y_n)^\top \in \mathbb{R}^n$ denoting a vector of realized observations of the variables $\mathcal{Y}_i$, the log-likelihood (with weights $a_i \geq 0$ for $i = 1,\ldots, n$) and its gradient are given as
\begin{alignat}{4}
 l(\eta(\theta))&= \sum_{i=1} ^n  a_i(y_i  \vartheta(\eta_i(\theta))-b(\vartheta(\eta_i(\theta)))) \quad \textnormal{and} \label{two}
 \\
 \nabla_\theta l(\eta(\theta))&=X^\top  u(\eta(\theta))\label{three},
\end{alignat}
respectively,  where $\vartheta : \mathbb{R} \to \mathbb{R}$ denotes
the canonical parameter function, and $u(\eta) \coloneqq\nabla_\eta
l(\eta) $ is the score statistic, see Appendix \ref{sec:expfam}.

 The main problem considered in this paper is the computation of the penalized maximum likelihood estimate (PMLE),
\begin{alignat}{4}\label{five}
  \theta^\ast\coloneqq\argmin_{\theta\in \RR^p}
  -l(\eta(\theta))+\lambda J(\theta), 
\end{alignat}
where $J:\RR^p \to (-\infty, \infty]$ is a proper, convex and closed
penalty function, and $\lambda \geq 0$ is a regularization parameter
controlling the amount of penalization. Note that $J$ is allowed to
take the value $\infty$, which can be used to enforce convex parameter
constraints. The objective function of this minimization problem is thus the penalized negative log-likelihood, denoted
\begin{alignat}{4}
F\coloneqq-l+\lambda J,
\end{alignat}
where $-l$ is continuously differentiable. 

For a GLAM the vector $y$ is assumed given as $y=\ve(Y)$ (the  $\ve$
operator is discussed in Appendix \ref{app:rho}), where 
$Y$ is an $n_1\times \cdots\times n_d$ $d$-dimensional array. The design matrix $X$ is assumed to be a 
concatenation of $c$ matrices 
\begin{alignat*}{4}
  X=[X_1\vert X_2\vert \ldots\vert X_c],
\end{alignat*}
where the $r$th component is a tensor product, 
\begin{alignat}{4}\label{fivenew}
 X_r=X_{r, d} \otimes X_{r,d - 1} \otimes \cdots \otimes X_{r,1},
  \end{alignat}
of $d$ matrices. The matrix $X_{r,j}$ is an $n_{j}\times p_{r,j}$ matrix, such that
  \begin{alignat*}{4}
n = \prod_{j=1}^{d} n_{j}, \quad p_r \coloneqq \prod_{j=1}^{d} p_{r,j}, \quad p = \sum_{r=1}^c p_r.  
\end{alignat*}
We let $\langle X_{r,j}\rangle\coloneqq \langle X_{1,1},\ldots ,X_{c,d} \rangle$ denote the tuple of marginal design matrices. 

The assumed data structure induces a corresponding structure on the parameter vector, $\theta$, as a concatenation of $c$  vectors,
\begin{alignat*}{4}
\theta^\top=(\ve(\Theta_1)^\top, \ldots,    \ve(\Theta_c)  ^\top ),
\end{alignat*}
with  $\Theta_r$ a $p_{r,1}\times\cdots\times p_{r,d}$  $d$-dimensional array. We let $\langle\Theta_{r}\rangle\coloneqq\langle\Theta_{1},\ldots,\Theta_{c}\rangle$ denote the tuple of parameter arrays.

Given this structure it is possible to define a map, $\rho$, such that
with $\theta_r = \ve(\Theta_r)$,
\begin{alignat}{4}\label{seventeen}
  X_r\theta_r  = \ve \big( \rho(X_{r,d},\ldots,
  \rho(X_{r,2},(\rho(X_{r,1},\Theta_r)))\ldots) \big)
\end{alignat}
for $r = 1, \ldots, c$. The algebraic details of $\rho$ are spelled
out in Appendix \ref{app:rho}.

As a consequence of the array structure, the linear predictor can be
computed using $\rho$ without explicitly constructing $X$. The most
obvious benefit is that no large tensor product matrix needs to be
computed and stored. In addition, the array structure can be
beneficial in terms of time complexity. As noted in \cite{buis1996},
with $X_{r,j}$ being a square $n_r \times n_r$ matrix, say, the
computation of the direct matrix-vector product in \eqref{seventeen} has $O(n_r^{2d})$ time
complexity, while the corresponding array computation has
$O(dn_r^{d+1})$ time complexity. This reduced time complexity for $d
\geq 2$ translates, as mentioned in the introduction, directly into
a computational advantage for computing the PMLE with a quadratic
penalty function, see  \cite{currie2006}. For non-quadratic penalty functions
the translation is less obvious, but we present one algorithm that is
capable of benefitting from the array structure.

\section{Penalized estimation  in a GLAM} \label{sec:alg} 

 In most situations the PMLE must be computed by an iterative algorithm.  We present an algorithm that solves the optimization problem \eqref{five}  by iteratively optimizing a partial quadratic approximation to the objective function while exploiting the array structure. The proposed algorithm is a combination of a gradient based descent (GD) algorithm with a proximal gradient (PG) algorithm. The resulting algorithm, which we call GD-PG, thus consists of the following two parts:
\begin{itemize}
\item an outer GLAM enhanced GD loop
\item an inner GLAM enhanced PG loop.
\end{itemize}
We present these two loops in the sections below postponing the details on how the array structure can be exploited to Section \ref{sec:glamops}, where it is explained in detail how the two loops can be enhanced for GLAMs. 

\subsection{The outer GD loop}\label{subsubsec:ipwls}
The outer loop consists of a sequence of descent steps based on a
partial quadratic approximation of the objective function. This
results in a sequence of estimates, each of which is defined in terms
of a penalized weighted least squares estimate and whose computation
involves an iterative choice of weights. The weights can be chosen so that the
inner loop can better exploit the array structure.

  For $k\in \NN$ and $\theta^{(k)}\in \RR^p$ let $\eta^{(k)} = \eta(\theta^{(k)})$ and $u^{(k)} = \nabla_\eta l(\eta^{(k)})$, let $W^{(k)}$ denote  a  positive  definite    diagonal $n\times n$  weight matrix and let 
$z^{(k)}$ denote the  $n$-dimensional vector (the working response) given by 
\begin{alignat}{4}\label{ten}
z^{(k)}\coloneqq(W^{(k)})^{-1}u^{(k)} + \eta^{(k)}.
\end{alignat}
 The sequence $(\theta^{(k)})$ is defined recursively from an initial $\theta^{(0)}$ as follows. Given $\theta^{(k)}$ let
\begin{alignat}{4}\label{eleven}
\tilde{\theta}^{(k+1)} \coloneqq \argmin_{\theta\in \mathbb{R}^p}  \frac{1}{2n}\Vert \sqrt{W^{(k)}}(X\theta-z^{(k)})\Vert^2_2+\lambda  J(\theta)
\end{alignat}
denote the penalized weighted least squares estimate and define 
\begin{alignat}{4}
\theta^{(k+1)} \coloneqq \theta^{(k)} + \alpha_k (\tilde{\theta}^{(k+1)} - \theta^{(k)}),
\end{alignat}
where the stepsize $\alpha_k > 0$ is determined to ensure sufficient
descent of the objective function, e.g. by using the Armijo rule. A
detailed convergence analysis is given in Section \ref{sec:conv},
where the relation to the class of gradient based descent
algorithms in \cite{tseng2009} is established.

\subsection{The inner PG loop}

The inner loop solves \eqref{eleven} by a proximal gradient algorithm. To formulate the algorithm consider 
a generic version of \eqref{eleven} given by
\begin{alignat}{4}\label{twentytwo}
x^\ast \coloneqq \argmin_{x\in \RR^p} h(x) + \lambda J(x),
\end{alignat}
where $h :\RR^p \to\RR$ is convex and continuously differentiable. It is assumed that there exists a minimizer $x^\ast$. Define for $\gamma > 0$ the proximal operator, $\mathrm{prox}_{\gamma}: \mathbb{R}^p \to \mathbb{R}^p$, by 
$$\mathrm{prox}_{\gamma}(z) = \argmin_{x \in \mathbb{R}^p}  \Big\{\frac{1}{2}\Vert x-z \Vert_{2}^2 + \gamma J(x)\Big\}.$$
The proximal operator is particularly easy to compute for a separable penalty function like the 1-norm or the squared 2-norm. 
Given a stepsize $\delta_k > 0$, initial values $x^{(0)} = x^{(1)} \in \mathbb{R}^p$ and an extrapolation sequence $(\omega_l)$ with $\omega_l \in [0,1)$ define the sequence $(x^{(l)})$ recursively by
\begin{alignat}{4}
 y & \coloneqq  x^{(l)} + \omega_{l} \Big(x^{(l)} - x^{(l-1)}\Big) \quad \textrm{and} \\
 x^{(l+1)} & \coloneqq \mathrm{prox}_{\delta_k \lambda}(y - \delta_k \nabla h(y)). 
\end{alignat}
The choice of $\omega_l = 0$ for all $l \in \mathbb{N}$ gives the
classical proximal gradient algorithm, see \cite{Parikh2014}. Other
choices of the extrapolation sequence, e.g. $\omega_l = (l - 1) / (l + 2)$,
can accelerate the convergence. Convergence results can be established if $\nabla h$ is Lipschitz continuous and $\delta_k$ is chosen sufficiently small -- see Section \ref{sec:conv} for further details.

For the convex function 
\begin{alignat}{4}\label{twentyfour}
 h(\theta) &\coloneqq\frac{1}{2n} \Vert  \sqrt{W^{(k)}}(X\theta- z^{(k)} )\Vert^2_2
\end{alignat}
we have that 
   \begin{alignat}{4}\label{twentyfive}
 \nabla h(\theta) =   \frac{1}{n}X^\top W^{(k)}( X\theta-z^{(k)} ).
\end{alignat}
This shows that $\nabla h(\theta)$ is Lipschitz continuous, and its explicit form in \eqref{twentyfive} indicates how the array structure can be exploited -- see also Section \ref{sec:glamops}.

\subsection{The GD-PG algorithm}\label{subsec:iwlsfpga}

The combined GD-PG algorithm is outlined as Algorithm \ref{alg:cgdfpg}
below. It is formulated using array versions of the model
components. Especially, $U^{(k)}$ and $Z^{(k)}$ denote
$n_1\times \cdots\times n_d$ array versions of the score statistic,
$u^{(k)}$, and the working response, $z^{(k)}$, respectively. Also $V^{(k)}$ is an $n_1\times \cdots\times n_d$ array
containing the diagonal of the $n\times n$ weight matrix
$W^{(k)}$. The details on how the steps in the algorithm can exploit
the array structure are given in Section \ref{sec:glamops}.
 
\begin{algorithm}[H] 
 \caption{GD-PG}
 \label{alg:cgdfpg}
\begin{algorithmic}[1]
\REQUIRE $\langle \Theta_r^{(0)} \rangle$, $\langle X_{r,j} \rangle$
 \FOR{$k=0$ to $K\in \NN$}
\STATE given $ \langle\Theta_r^{(k)}\rangle$: compute $U^{(k)}$, specify $V^{(k)}$  and compute $Z^{(k)}$
\STATE specify the proximal stepsize $\delta_k$
\STATE given  $ \langle\Theta_r^{(k)}\rangle, V^{(k)}, Z^{(k)}, \delta_k$:  compute  $\langle \tilde\Theta_r^{(k+1)}\rangle$ by the inner PG loop
\STATE given $ \langle\Theta_r^{(k)}\rangle, \langle\tilde\Theta_r^{(k+1)}\rangle$: use a line search to compute $\langle\Theta_r^{(k+1)}\rangle$
\IF{convergence criterion is satisfied} 
\STATE break
\ENDIF
\ENDFOR
\end{algorithmic}
\end{algorithm}

The outline of Algorithm \ref{alg:cgdfpg} leaves out some details that
are required for an implementation. In step 2 the weights must be
specified. In Section \ref{sec:conv} we present results on convergence
of the outer loop, which put some restrictions on the choice of
weights. In step 3 the proximal gradient stepsize must
be specified. In Section \ref{sec:conv} we give a computable upper
bound on the stepsize that ensures convergence of the inner PG
loop. Convergence with the same convergence rate can also be ensured
for larger stepsizes if a backtracking step is added to the inner PG
loop. In step 4, $\langle\Theta_r^{(k)}\rangle$ is a natural choice of
initial value in the inner PG loop, but this choice is not necessary
to ensure convergence. In step 4 it is, in addition, necessary to
specify the extrapolation sequence. Finally, in step 5 a line search
is required. In Section \ref{sec:conv} convergence of the
outer loop is treated when the Armijo rule is used.

\section{Applications to multidimensional smoothing}\label{sec:multdimsmooth}

As a main application of the GD-PG algorithm we consider multidimensional smoothing, which can be formulated in the framework of GLAMs by using a basis expansion with tensor product basis functions. We present the framework below and report the results obtained for two real data sets. 

\subsection{A generalized linear array model for smoothing}\label{subsec:gam}

Letting $\X_{1},\ldots,\X_d \subseteq \mathbb{R}$ denote $d$ finite sets define the $d$-dimensional  grid
\begin{alignat*}{4}
\mathcal{G}_d\coloneqq  \X_1\times\cdots\times\X_d.
\end{alignat*}
The set $\X_j$ is the  set of (marginal) grid points in the $j$th
dimension and $n_j\coloneqq\vert\X_{j}\vert $  denotes the number of
such marginal points in the $j$th dimension. We have a total of
$n\coloneqq\prod_{j = 1}^d n_j$ $d$-dimensional joint grid points,  or  $d$-tuples,
\begin{alignat*}{4}
(x_1,\ldots,x_d) \in \mathcal{G}_d.
\end{alignat*}

For each of the $n$ grid points  we observe a corresponding  grid
value $y_{x_1,\ldots,x_d}\in \RR$ assumed to be a realization of a
real valued random variable $\mathcal{Y}_{x_1,\ldots,x_d}$ with finite
mean. That is, the observations can be regarded as a $d$-dimensional
array $Y$. With $g:\RR\to\RR$ a link function let
\begin{alignat}{4}\label{thirtyfour}
f(x_1,\ldots,x_d) \coloneqq g(\Exp(\mathcal{Y}_{x_1,\ldots,x_d})),\quad  (x_1,\ldots,x_d)\in \G_d.
\end{alignat}
The objective is to estimate $f$, which is assumed to possess some form of regularity as a function of $(x_1,\ldots,x_d)$. Assuming that $f$ belongs to the span of $p$ basis functions, $\phi_1, \ldots, \phi_p$, it holds that
 \begin{alignat*}{4}
f(x_1,\ldots,x_d)=\sum_{m=1}^{p} \beta_m\phi_m(x_1,\ldots,x_d), \quad  (x_1,\ldots,x_d)\in \G_d,
\end{alignat*}
for $\beta \in \mathbb{R}^p$. If the basis function evaluations are
collected into an $n\times p$ matrix
$\Phi\coloneqq(\phi_m((x_1,\ldots,x_d)_i))_{i,m}$, and if the entries
in the array $Y$ are realizations of independent random variables from
an exponential family as described in Appendix \ref{sec:expfam}, the resulting model is a GLM with design matrix $\Phi$ and regression coefficients $\beta$.

For $d\geq 2$ the $d$-variate basis functions can be specified via a tensor product construction in terms of $d$ (marginal) sets of univariate functions 
 by 
\begin{alignat}{4}
 \phi_{m_1,\ldots,m_d} \coloneqq   \phi_{1,m_1} \otimes \phi_{2,m_2} \otimes \cdots \otimes \phi_{d,m_d},  
\end{alignat}
where $\phi_{j,m} : \mathbb{R} \to \mathbb{R}$ for $j = 1, \ldots, d$ and $m = 1, \ldots, p_j$. 
The evaluation of each of the $ p_j$ univariate functions in the $n_j$
points in $\X_j$ results in an $n_j\times p_j$ matrix  $\Phi_j =
(\phi_{j, m}(x_k))_{k,m}$. It then follows that the $n\times p$  ($p\coloneqq\prod_{j=1}^d p_j$) tensor product matrix 
\begin{alignat}{4}
  \Phi = \Phi_{d}\otimes\cdots\otimes \Phi_{1}
\end{alignat} 
is identical to the design matrix for the basis evaluation in the tensor product basis, and the GLM has the structure required of a GLAM. 
 
\subsection{Benchmarking on real data}\label{subsec:data}

The multidimensional smoothing model described in the previous section was fitted using an $\ell_1$-penalized B-spline basis expansion to two real data sets using the GD-PG algorithm as implemented in the R package \verb+glamlasso+. See Section \ref{sec:glamlasso} for details about the R package. In this section we report benchmark results for \verb+glamlasso+ and the coordinate descent based implementation in the R package \verb+glmnet+, see \cite{friedman2010}.

For both data sets we fitted a sequence of models to data from an
increasing subset of grid points, which correspond to a sequence of
design matrices of increasing size. For each design matrix we fitted
100 models for a decreasing sequence of values of the penalty
parameter $\lambda$. We report the run time for fitting the sequence
of 100 models using \verb+glamlasso+ and \verb+glmnet+. We also report
the run time for the combined computation of the tensor product design
matrix and the fit using \verb+glmnet+. The latter is more relevant
for a direct comparison with \verb+glamlasso+, since \verb+glamlasso+
requires only the marginal design matrices while \verb+glmnet+
requires the full tensor product design matrix.

To justify the comparison we report the relative deviation of the objective function values attained by \verb+glamlasso+ from the objective function values attained by \verb+glmnet+, that is,
\begin{alignat}{4}\label{thirtythree}
\frac{F(\hat{\theta}^{\texttt{glamlasso}})-F(\hat{\theta}^{\texttt{glmnet}})}{\vert F(\hat{\theta}^{\texttt{glmnet}})\vert }
\end{alignat}
with  $\hat{\theta}^{\texttt{x}}$ denoting the estimate computed by method \verb+x+. This ratio is computed for each fitted model. We note that \eqref{thirtythree} has a tendency to blow up in absolute value when $F$ becomes small, which happens for small values of $\lambda$. 

The benchmark computations were carried out on a  Macbook Pro   with a  2.8 GHz Intel core i7 processor and 16 GB  of 1600 MHz DDR3  memory.   Scripts and data are included as supplemental materials online.

\subsubsection{Gaussian neuron data}\label{subsubsec:neuro}
The first data set considered consists of spatio-temporal voltage sensitive dye  recordings of a ferret brain provided by Professor Per Ebbe Roland, see \cite{roland2006}. The data set consists of images of size $25\times 25 $ pixels recorded with a time resolution of 0.6136 ms per image. The images were recorded over 600 ms, hence  the total size of this 3-dimensional array data set is  $25\times25\times 977$ corresponding to $n = 610,625$ data points. 

As basis functions we used cubic B-splines with
$p_j\coloneqq\max\{[n_j/5],5\}$  basis functions in each dimension
(see \cite{currie2006} or \cite{wood2006}). This corresponds to a
parameter array of size  $5\times 5\times 196$ ($p=4,900$) and a
design matrix of size $610,625\times 4,900$ for the entire data
set. The byte size for representing this design matrix as a dense matrix was approximately 22 GB. For
the benchmark we fitted Gaussian models with the identity link
function to the full data set as well as to subsets of the data set
that correspond to smaller design matrices. 

Figure \ref{fig:neuronex} shows an example of the raw data and the smoothed fit for a particular time point. Movies of the raw data and the smoothed fit can be found as supplementary material. 

Run times and relative deviations are shown in Figure
\ref{fig:runneuron}. The model could not be fitted using \verb+glmnet+
to the full data set due to the large size of the design matrix, and
results for \verb+glmnet+ are thus only reported for models that could
be fitted. The run times for \verb+glamlasso+ were generally smaller
than for \verb+glmnet+, and were, in particular, relatively
insensitive to the size of the design matrix. When a sparse matrix
representation of the design matrix was used, \verb+glmnet+ was able
to scale to larger design matrices, but it was still clearly outperformed by
\verb+glamlasso+ in terms of run time. The relative deviations
in the attained objective function values were quite small.

\begin{figure}[t]
\begin{center}
\includegraphics[scale=0.35]{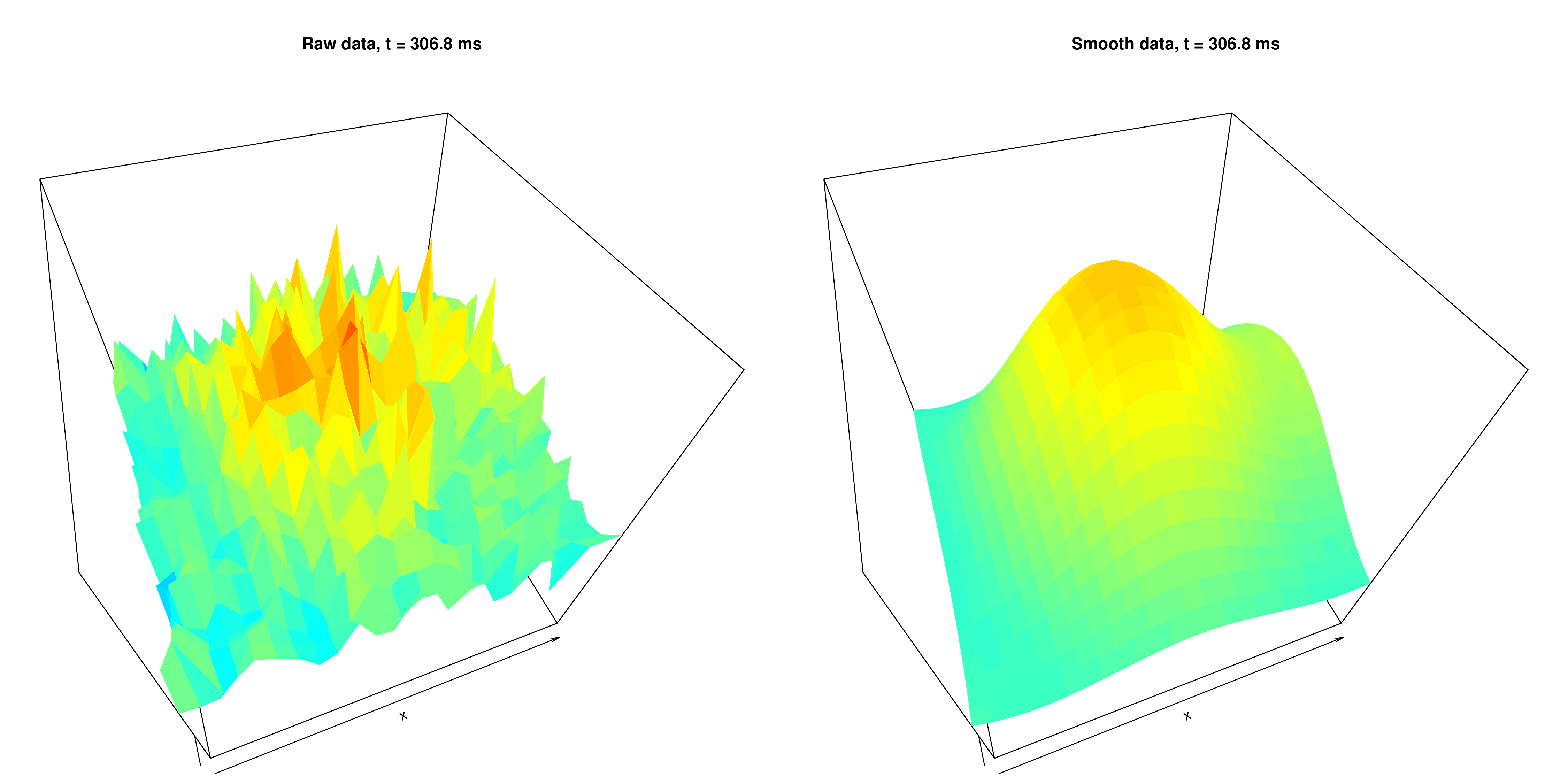}
\caption{The raw neuron data (left) and the smoothed fit (right) after 306.8 ms. The supplementary material contains movies of the complete raw data and smoothed fit.}
\label{fig:neuronex}
\end{center}
\end{figure}

\begin{figure}[t]
\begin{center}
\begin{tabular}{cc}
\includegraphics[scale = 0.45]{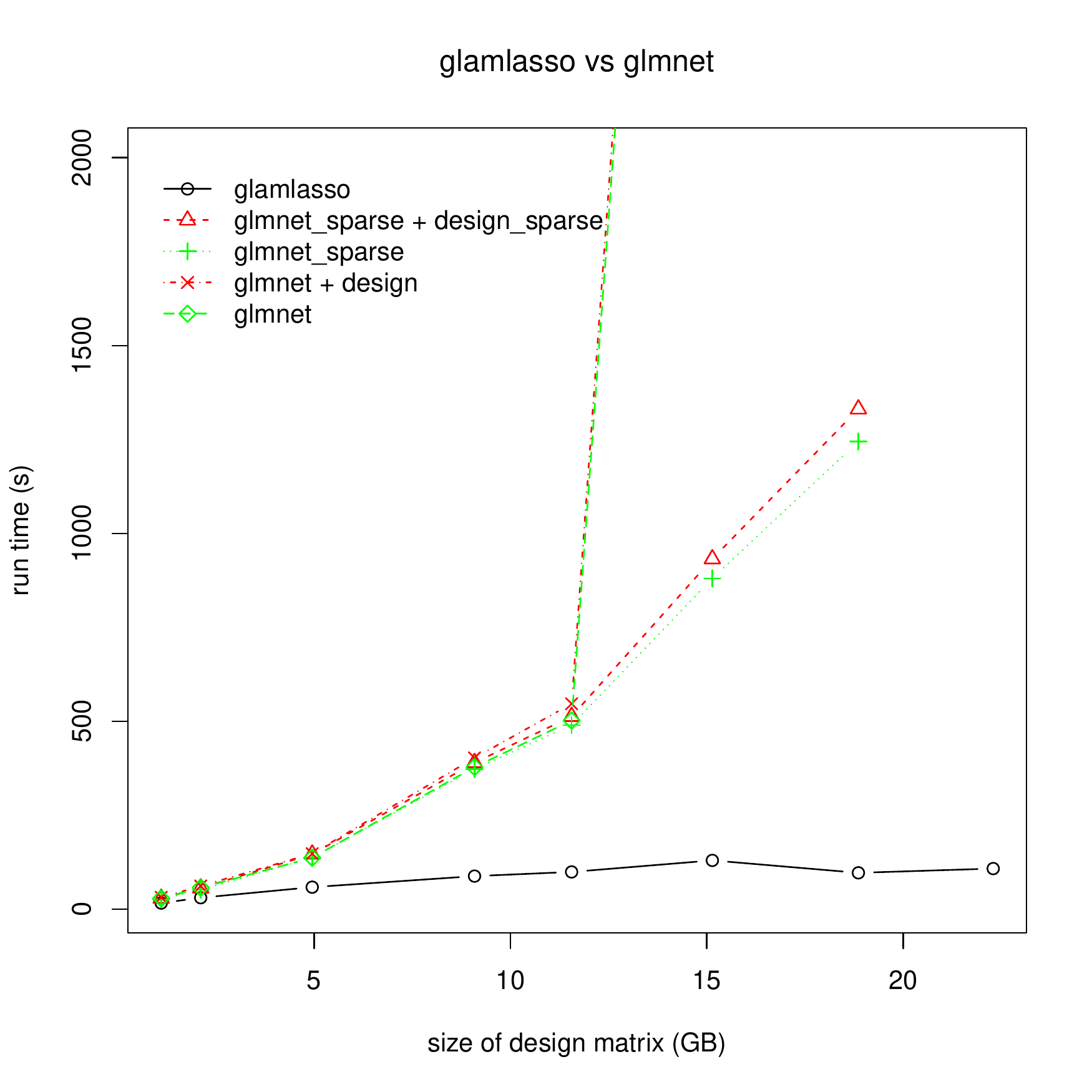}
\includegraphics[scale = 0.45]{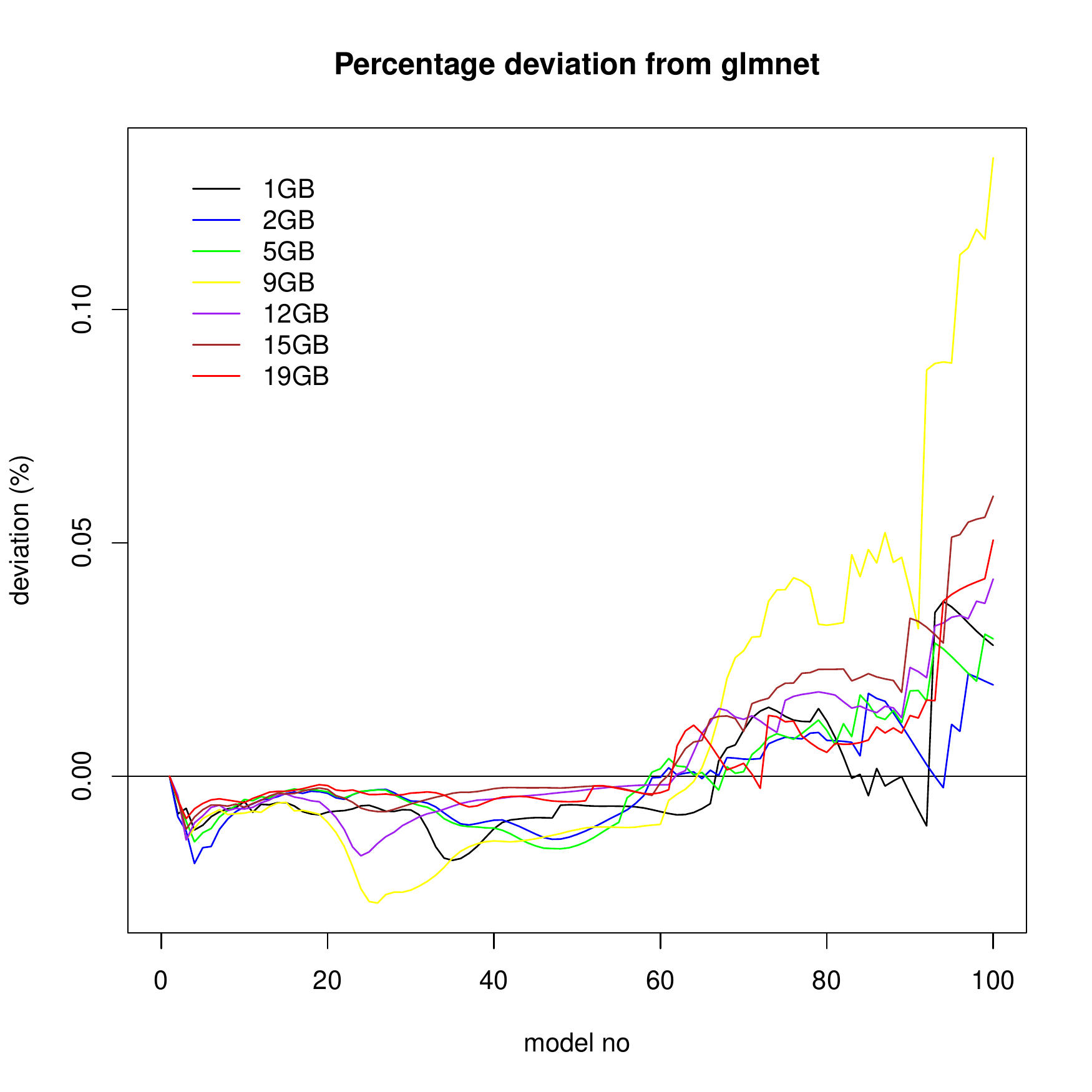} &
\end{tabular}
 \caption{Benchmark results for the neuron data. Run time in seconds is shown as a function of the size of the design matrix, when not stored in sparse format,  in GB (left). Relative deviation in the attained objective function values as given by \eqref{thirtythree} is shown as a function of model number (right), where a larger model number corresponds to less penalization (smaller $\lambda$).} 
\label{fig:runneuron}
\end{center}
\end{figure}

\begin{figure}[t]
\begin{center}
\includegraphics[scale = 0.49]{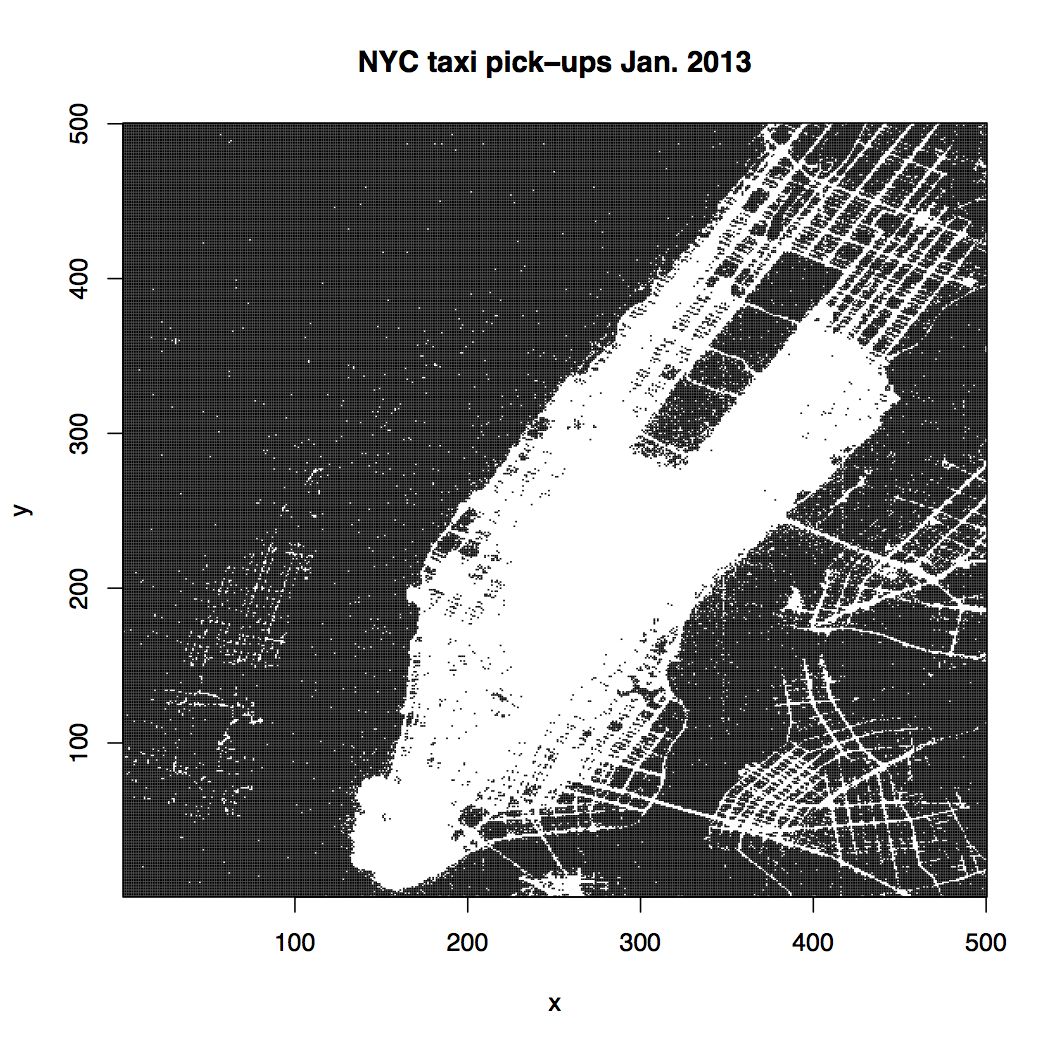}
\includegraphics[scale = 0.5]{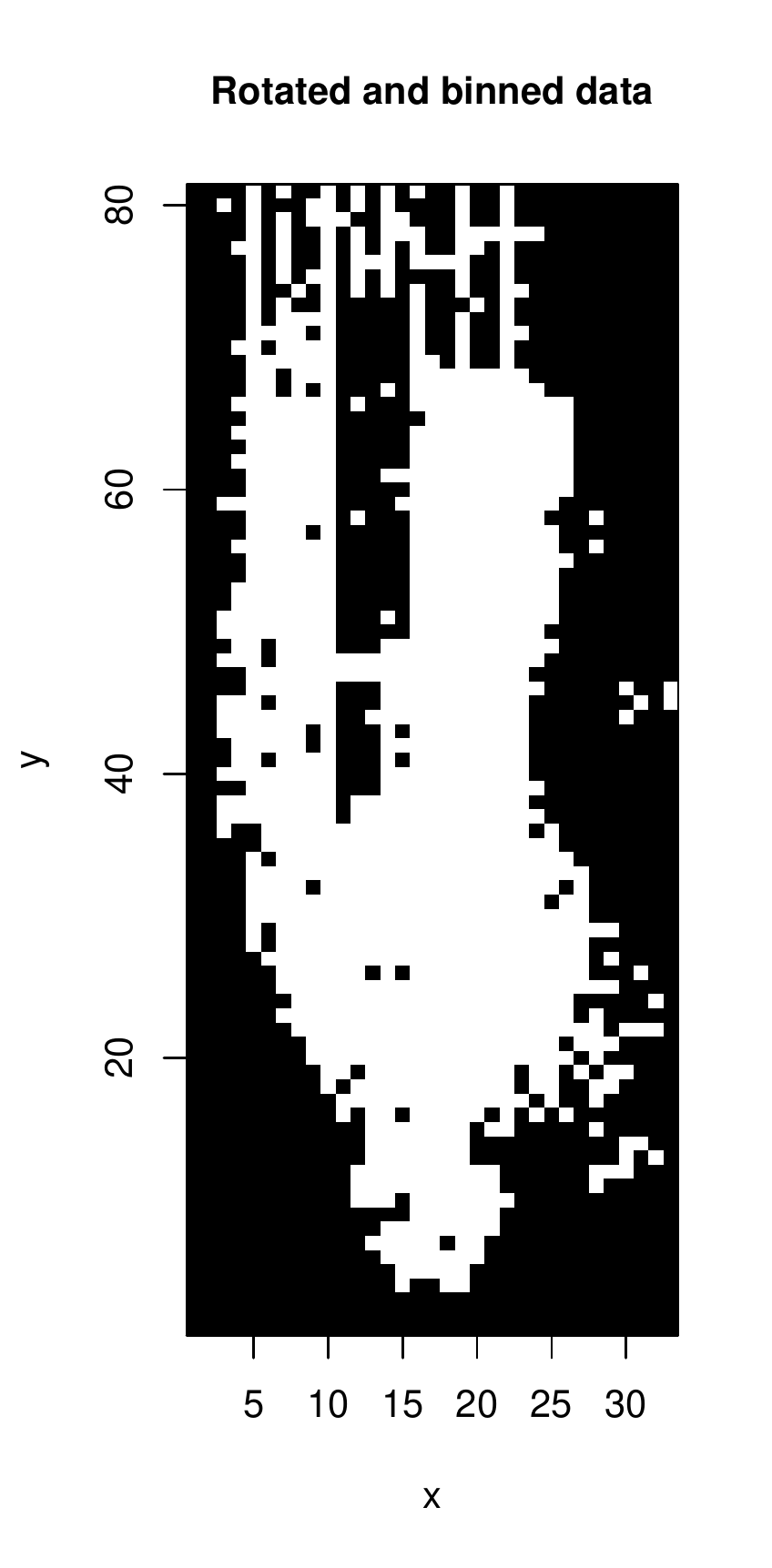}
\caption{Binned counts of registered NYC taxi pickups for January 2013 using $500 \times 500$ spatial bins (left) and the same data rotated, binned to $100 \times 100$ spatial bins and cropped to cover Manhattan only (right).}
\label{fig:five}
\end{center}
\end{figure}

\subsubsection{Poisson taxi data}\label{subsubsec:taxi}
 The second data set considered consists of spatio-temporal information on registered taxi pickups in New York City during January 2013. The data can be  download from the webpage \verb+www.andresmh.com/nyctaxitrips/+. We used a subset of this data set consisting of  triples containing  longitude, latitude and  date-time of the pickup.  First we cropped the data to pickups with longitude  in $[-74.05^\circ, -73.93^\circ]$ and latitude  in $[40.7^\circ,40.82^\circ]$.  Figure \ref{fig:five} shows the binned counts of all pickups during January 2013 with 500 bins in each spatial dimension. Pickups registered in  Hudson  or  East River were ascribed to noise in the GPS recordings.

For this example attention was restricted to Manhattan pickups during
the first week of January 2013. To this end the data was rotated and
summarized as binned counts in $100 \times 100 \times 168$
spatial-temporal bins. Each temporal bin represents one hour. The data
was then further cropped to cover Manhattan only, which removed the
large black parts -- as seen on Figure \ref{fig:five} -- where pickups were rare. The total size of the data set was $33\times 81\times 168$ corresponding to $n =  449,064 $ data points. The observation in each bin consisted of the integer number of pickups registered in that bin. 

 \begin{figure}[t]
\begin{center}
\includegraphics[scale=0.4]{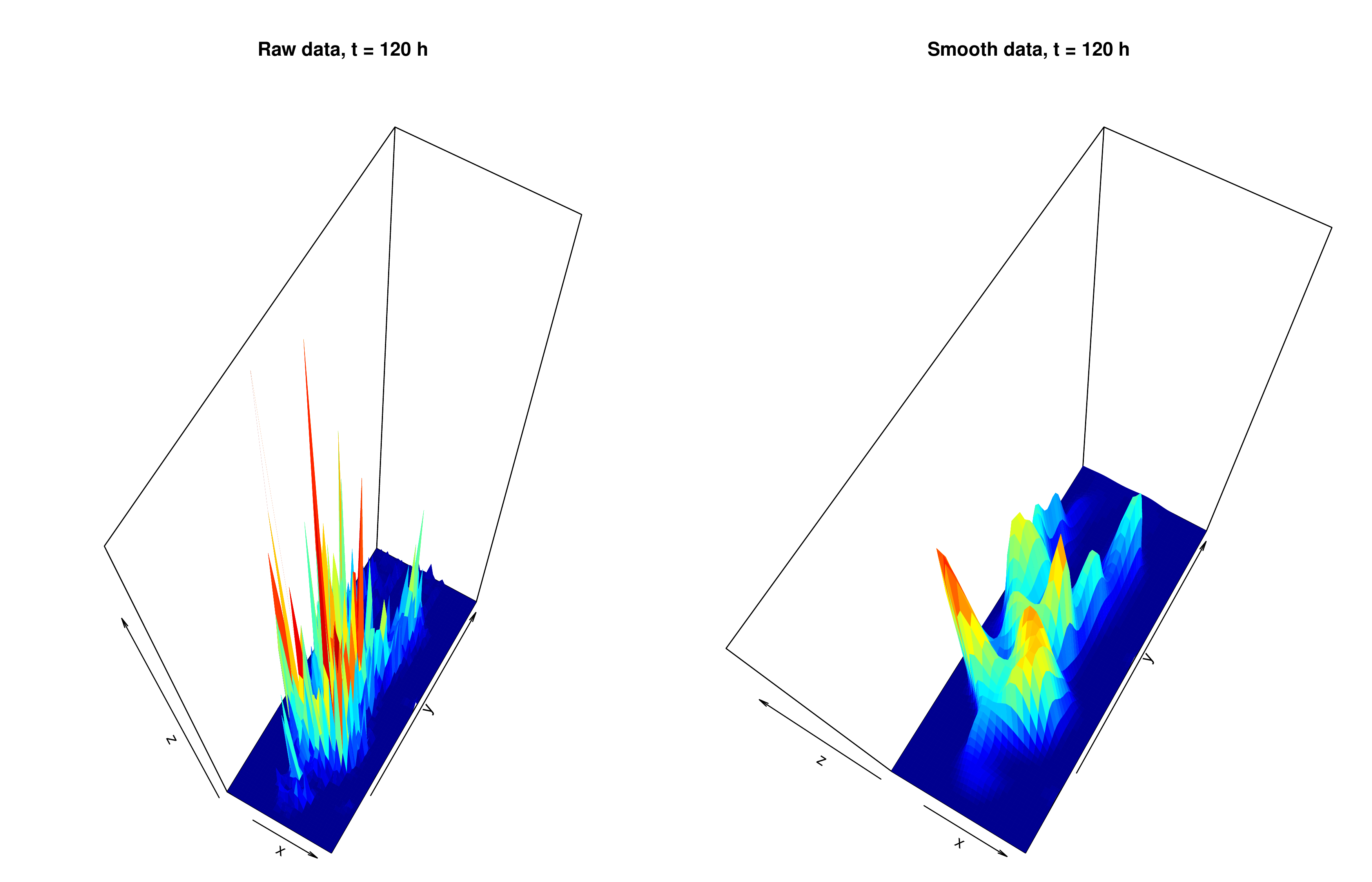}
\caption{The raw NYC taxi data (left) and the smoothed fit (right) around midnight on Saturday, January 5, 2013. The supplementary material contains movies of the complete raw data and smoothed fit.}
\label{fig:taxiex}
\end{center}
\end{figure}

\begin{figure}[h!]
\begin{center}
\begin{tabular}{cc}
\includegraphics[scale = 0.45]{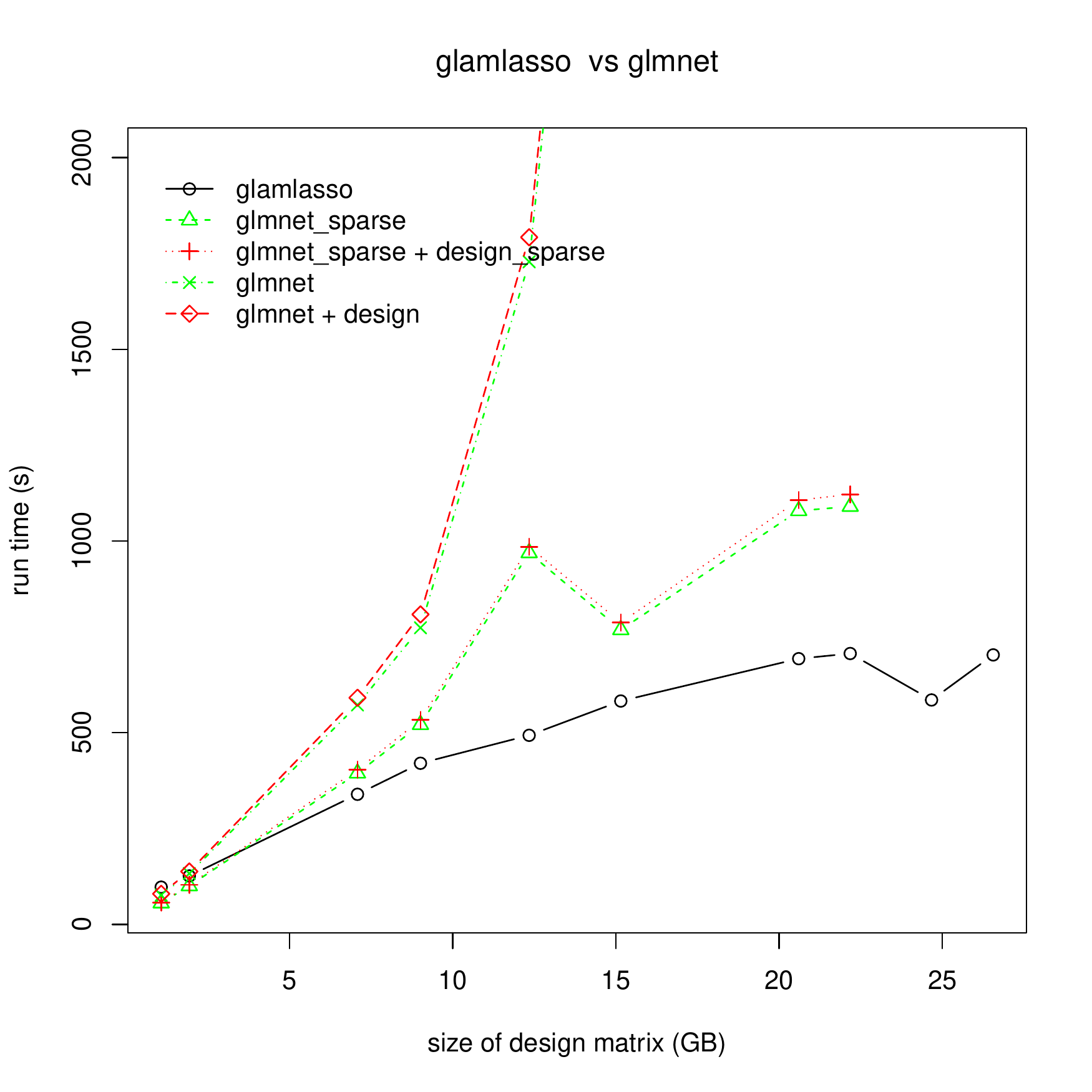} &
\includegraphics[scale = 0.45]{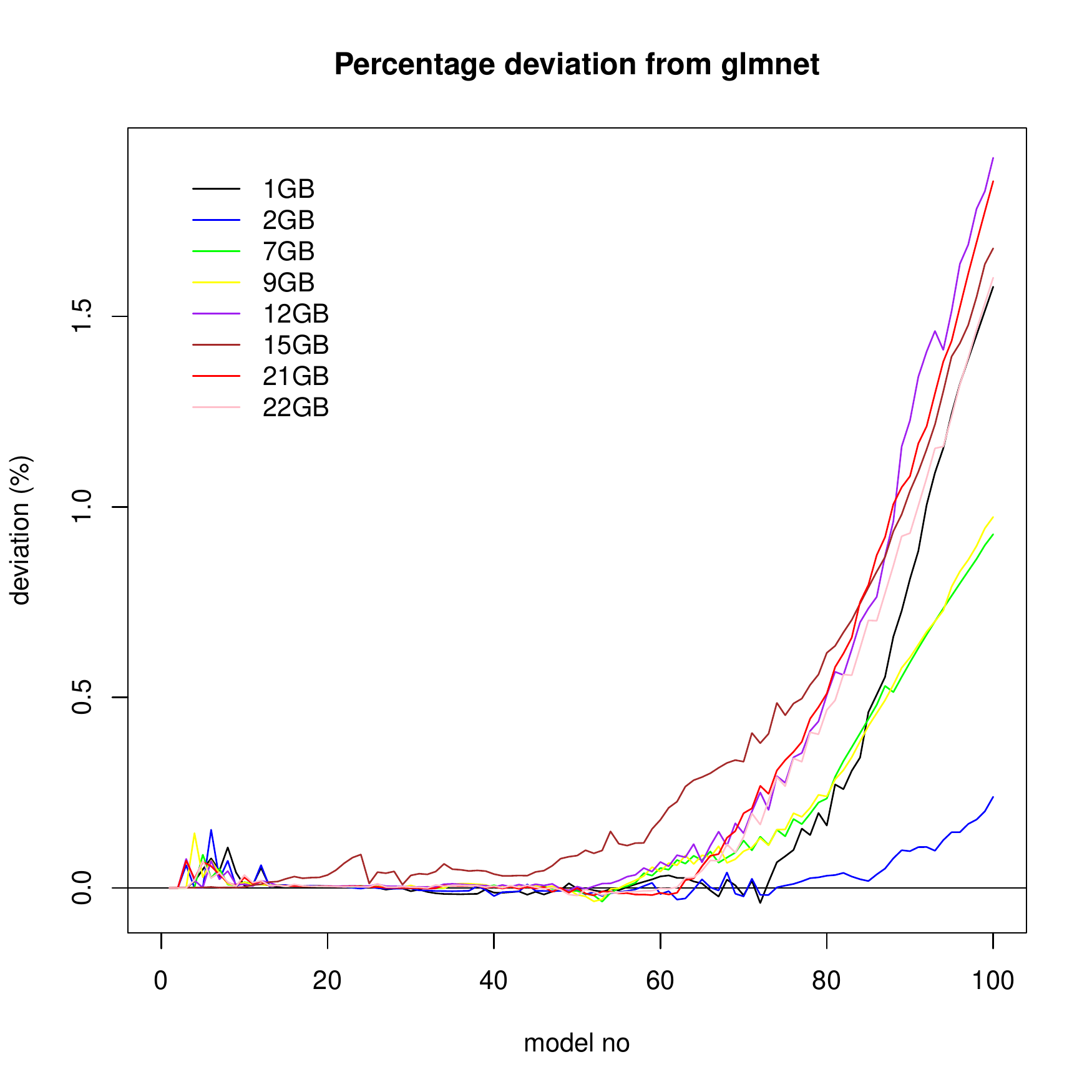}
\end{tabular}
\caption{Benchmark results for the taxi data. Run time in seconds is shown as a function of the size of the design matrix, when not stored in sparse format, in GB (left). Relative deviations in the attained objective function values as given by \eqref{thirtythree} is shown as a function of model number (right), where a larger model number corresponds to less penalization (smaller $\lambda$).}
\label{fig:runtaxi}
\end{center}
\end{figure}

We  used $p_j\coloneqq\max\{[n_j/4],5\}$ cubic B-spline basis functions in each dimension. The resulting parameter array was $9\times 21\times 42$ corresponding to  $p = 7,938$ and a design matrix of size  $449,064\times 7,938$ for the entire data set. The byte size for representing this design matrix as a dense matrix was approximately 27 GB.  For
the benchmark we fitted Poisson models with the log link function to the full data set as well as to subsets of the data set
that correspond to smaller design matrices.   

Figure \ref{fig:taxiex} shows an example of the raw data and the smoothed fit for around midnight on Saturday, January 5, 2013. Movies of the raw data and the smoothed fit can be found as supplementary material.  

Run times and relative deviations are shown in Figure
\ref{fig:runtaxi}. As for the neuron data, the model could not be
fitted to the full data set using \verb+glmnet+, and results for
\verb+glmnet+ are only reported for models that could be
fitted. Except for the smallest design matrix the run times for
\verb+glamlasso+ were smaller than for \verb+glmnet+, and they appear
to scale better with the size of the design matrix. This was
particularly so when the dense matrix representation was used
with \verb+glmnet+. The design matrix was very sparse in this example,
and \verb+glmnet+ benefitted considerable in terms of run time from
using a sparse storage format. The relative
deviations in the attained objective function values were still
acceptably small though the values attained by \verb+glamlasso+ were
up to 1.5\% larger than those attained by \verb+glmnet+ for the
least penalized models (models fitted with small values of $\lambda$).

\begin{figure}[t]
\begin{center}
\includegraphics[scale = 0.45]{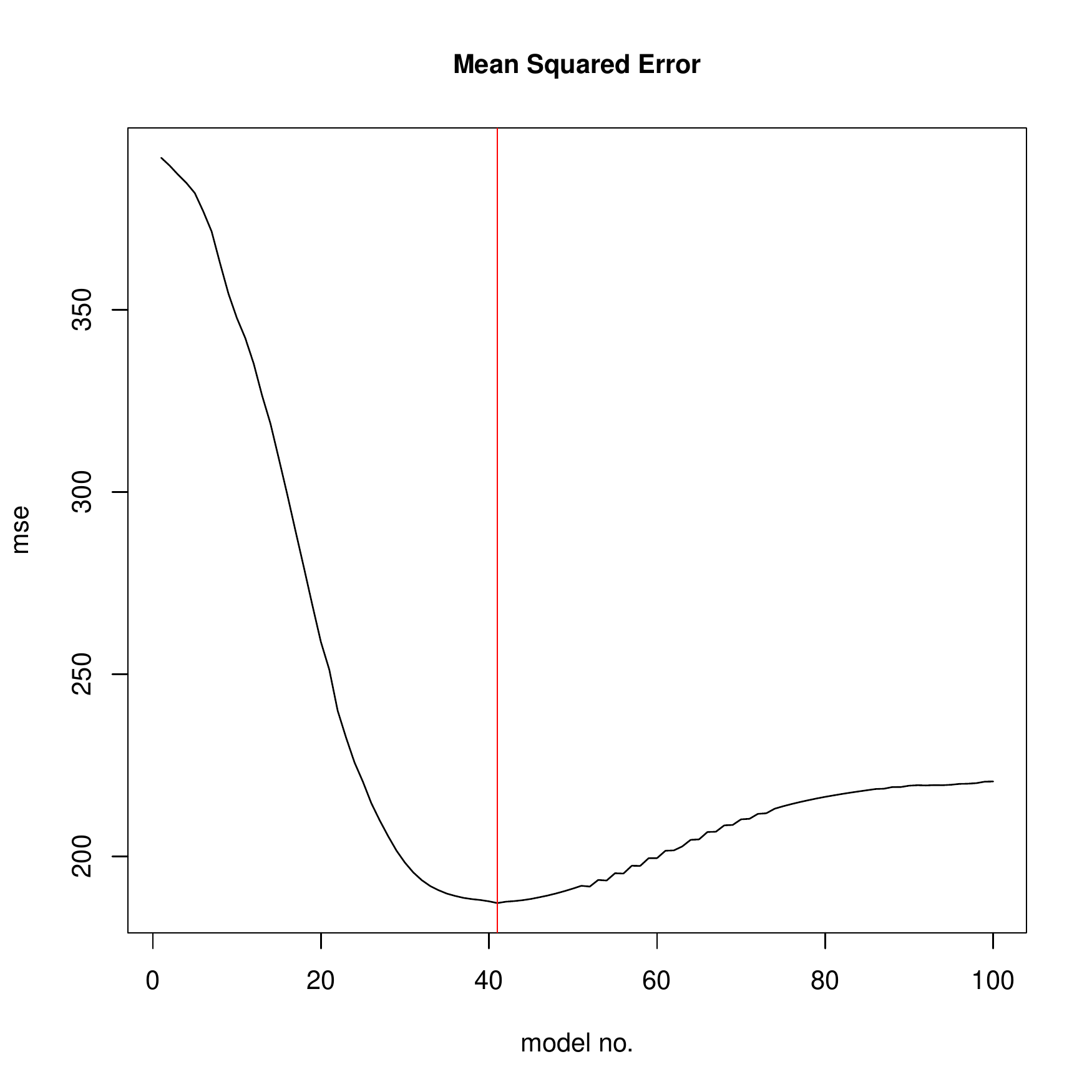}
\caption{The mean squared error for prediction on grid points left out of the model fitting as a function of model number. The vertical red line indicates the model with minimal MSE (model 41).}
\label{fig:eight}
\end{center}
\end{figure}

\subsection{Using incomplete array data}
\label{sec:incomplete}
The implementation in \verb+glamlasso+ allows for incompletely
observed arrays. This can, of course, be used for prediction of the
unobserved entries by computing the smoothed fit to the incompletely
observed array. In this section we show how it can also be used for selection of the tuning parameter $\lambda$.   We also refer to the  supplemental materials online for scripts and data.

\begin{figure}[t]
\begin{center}
\includegraphics[scale = 0.45]{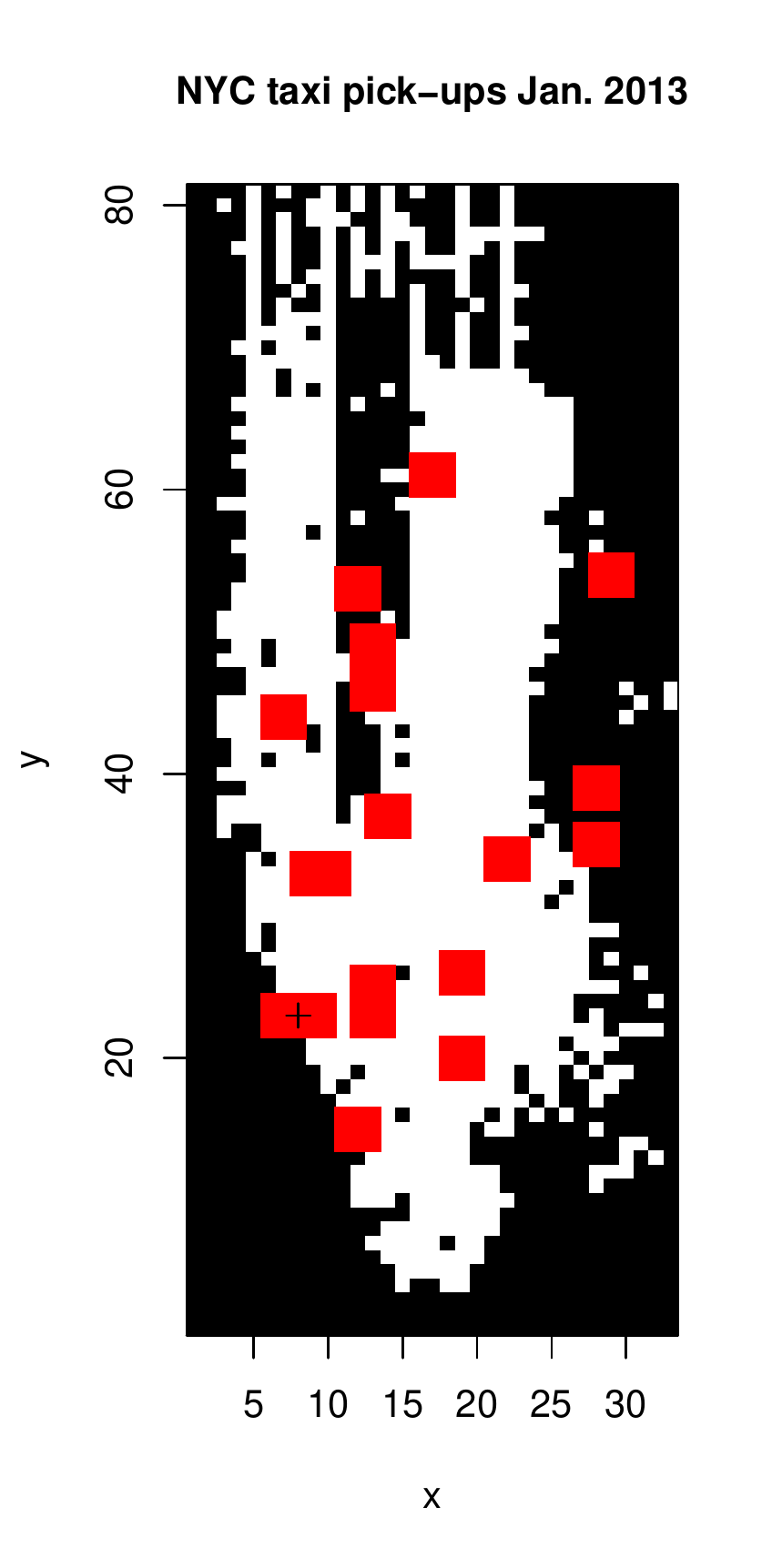}
\includegraphics[scale = 0.45]{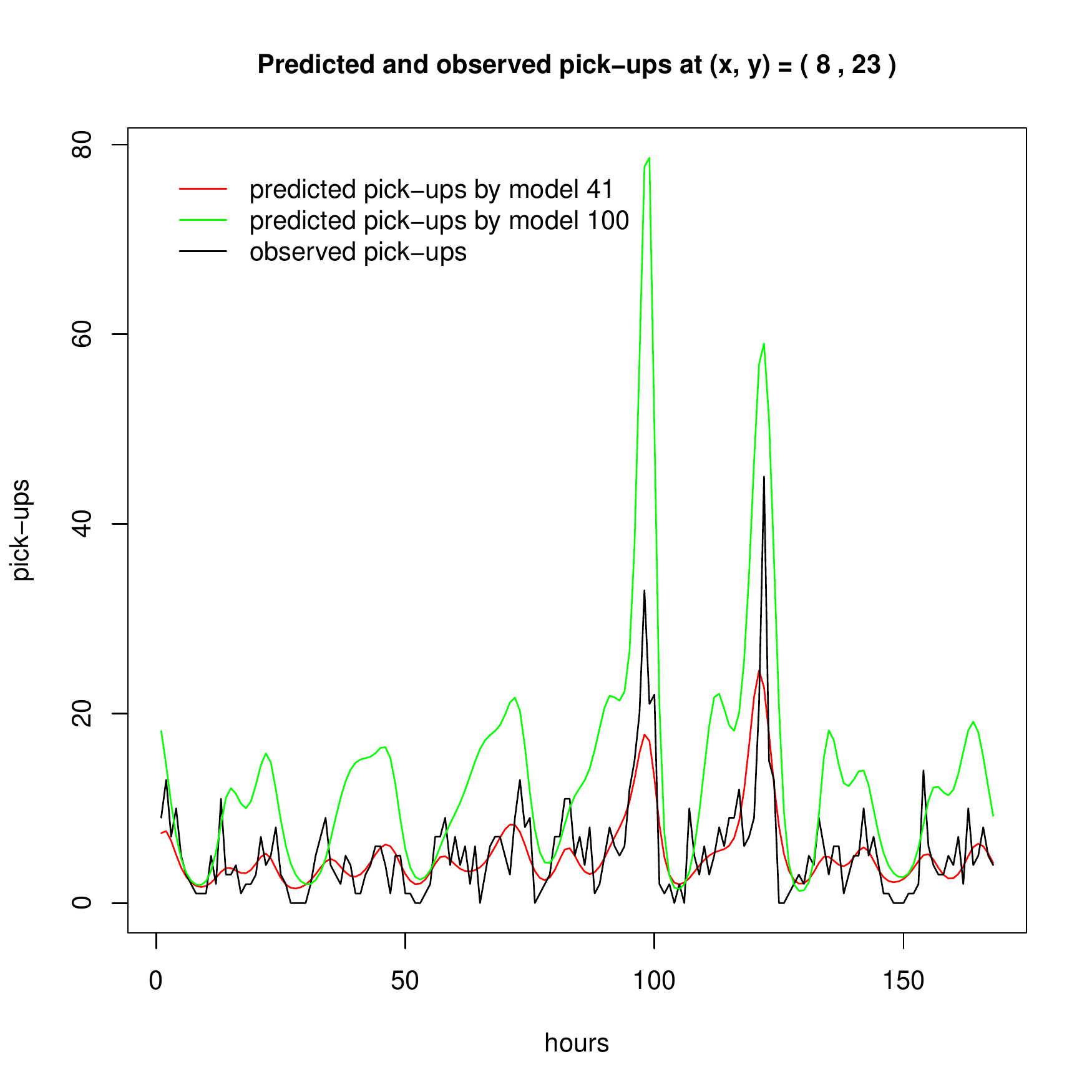}
\caption{Binned number of NYC taxi pickups as in Figure \ref{fig:five} (left) with red $3\times3$ squares indicating bins that were removed from the data before model fitting. Predicted and observed number of pickups at spatial bin $(8,23)$ (indicated with a ``+'' on the left figure) are shown as a function of time in hours (right). Model 100 predictions (green) were from the least penalized model while Model 41 predictions (red) were from the model with an overall minimal MSE.}
\label{fig:nine}
\end{center}
\end{figure}

We used the NYC taxi data and removed the observations for 19 randomly
chosen $3\times 3$ blocks of spatial bins (due to overlap of some of
the blocks this corresponded to 159 spatial bins). When fitting the model
using \verb+glamlasso+ the incompleteness is incorporated by setting
the weights corresponding to the missing values equal to zero for all
time points. We denote by $D$ the set of grid points that correspond
to the removed bins as illustrated by the red blocks in Figure
\ref{fig:nine}.

From \verb+glamlasso+ we computed a sequence of model fits corresponding to 100 values of $\lambda$, and for each value of $\lambda$ we computed the fitted complete array $\hat Y^{(\lambda)}$ and then the mean squared error (MSE), 
\begin{alignat*}{4}
\text{MSE}(\lambda)=\sum_{x \in D}(\hat Y^{(\lambda)}_x-Y_x)^2,
\end{alignat*}
as a function of $\lambda$, see Figure \ref{fig:eight}. Model 41 attained the overall minimal MSE. 

Figure \ref{fig:nine} shows predictions for one spatial bin. The under-smoothed Model 100 gives a poor prediction while the overall optimal Model 41 gives a much better prediction. 

\section{Convergence analysis}\label{sec:conv}

Our proposed GD-PG algorithm is composed of well known components, whose convergence properties have been extensively studied. We do, however, want to clarify under which conditions the algorithm can be shown to converge and in what sense it converges. The main result in this section is a computable upper bound of the step-size, $\delta_k$, in the inner PG loop that ensures convergence in this loop. This result hinges on the tensor product structure of the design matrix. 

We first state a theorem, which follows directly from \cite{Beck2010},
and which for a specific choice of extrapolation sequence gives the convergence rate for the inner PG loop for minimizing the objective function 
\begin{alignat}{4}
G \coloneqq  h + \lambda J,
\end{alignat}
where $h$ is given by \eqref{twentyfour}. In the following, $\Vert
A \Vert_{2}$ denotes the spectral norm of $A$, 
which is the largest singular value of $A$.
 
\begin{thm}\label{thm1} Let $x^\ast = \tilde{\theta}^{(k + 1)}$ denote
  the minimizer defined by \eqref{eleven} and let the extrapolation sequence for the
  inner PG loop be given by $\omega_l = (l - 1) / (l + 2)$. 
Let $(x^{(l)})$ denote the sequence obtained from the inner PG
loop. If $\delta^{(k)} \in (0, 1/L^{(k)}]$ where
\begin{alignat}{4} 
L^{(k)} :=\Vert X^{\top} W^{(k)}X \Vert_{2}/n
\label{eq:lip}
\end{alignat}
then
\begin{alignat}{4}
G(x^{(l)}) - G(x^\ast) \leq \frac{2 L_h^{(k)} \Vert x^{(0)} - x^* \Vert_2^2}{(l +
  1)^2}.
\label{eq:rate}
\end{alignat}
\end{thm}
\begin{proof} The theorem is a consequence of Theorem 1.4 in
  \cite{Beck2010} once we establish that $L^{(k)}$ is a Lipschitz
  constant for $\nabla h$. To this end note that the spectral norm $\Vert
\cdot \Vert_{2}$ is the operator norm induced by the
2-norm on $\mathbb{R}^p$, which implies that 
\begin{alignat}{4}\label{eighteennew}
\Vert   \nabla h( \theta)-\nabla h(\theta')\Vert_2&\leq \frac{1}{n} \Vert      X^{\top} W^{(k)}X \Vert_{2} \Vert \theta-\theta'\Vert_2,
\end{alignat}
and $L^{(k)}$ is indeed the minimal Lipschitz constant. It should be
noted that Theorem 1.4 in  \cite{Beck2010} is phrased in terms of an
acceleration sequence of the form $\omega_l = (t_l - 1) / t_{l+1}$ where $(t_l)$
is a specific sequence that fulfills $t_l \geq (l+1) / 2$. The
acceleration sequence considered here corresponds to $t_l = (l + 1) /
2$, and their proof carries over to this case without changes. 
\end{proof}

From \eqref{eq:rate} we see that the objective function values
converge at rate $O(l^{-2})$ for the given choice of extrapolation
sequence. Without extrapolation, that is, with $\omega_l = 0$ for all $l
\in \mathbb{N}$, the convergence rate is $O(l^{-1})$, see
e.g. Theorem 1.1 in \cite{Beck2010}. In this case $(x^{(l)})$
always converges towards a minimizer, see Theorem 1.2 in
\cite{Beck2010}. We are not aware of results that establish
convergence of $(x^{(l)})$ for general $h$ when extrapolation is
used. However, if $X$ has rank $p$ and the weights are all
strictly positive, the quadratic $h$ given by \eqref{twentyfour}
results in a strictly convex and level bounded objective function $G$, in
which case \eqref{eq:rate} forces
$(x^{(l)})$ to converge towards the unique minimizer. 

The following result shows how the tensor product structure can be
exploited to give a computable upper bound on the Lipschitz constant
\eqref{eq:lip}. 

\begin{prop}\label{prop:one} Let $W^{(k)}$ denote the diagonal weight matrix
  with diagonal elements $w_i^{(k)}$, $i = 1, \ldots, n$, then 
\begin{alignat}{4}\label{nineteennew} 
L^{(k)} \leq \hat{L}^{(k)} \coloneqq \frac{\max(w_{i}^{(k)}) }{n}\sum_{r=1}^c \prod_{j=1}^d\varrho(X_{r,j}^\top X_{r,j}) 
\end{alignat}
 where $\varrho$  denotes the spectral radius. 
\end{prop}
\begin{proof} Since the spectral norm is an operator norm it is
  submultiplicative, which gives that 
\begin{alignat}{4}\label{twentysix}
 L^{(k)}&\leq \frac{1}{n}\Vert      X^{\top}\Vert_{2} \Vert X\Vert_{2} \Vert W^{(k)} \Vert_{2} =\frac{1}{n}\Vert X\Vert_{2}^2 \Vert W^{(k)} \Vert_{2}. 
\end{alignat}
Now $W^{(k)}$ is diagonal with nonnegative entries, so $\Vert W^{(k)}
\Vert_{2} = \max(w_{i}^{(k)})$, and $\Vert X\Vert_{2}^2$ is the largest
eigenvalue of the symmetric matrix $X^{\top} X$ (the spectral radius), hence 
\begin{alignat}{4}
 L^{(k)} \leq \frac{\max(w_{i}^{(k)})}{n} \varrho(X^{\top} X).
\end{alignat}
Furthermore,  as $X^\top X$ is a  positive semidefinite matrix with
diagonal blocks given by $X_r^\top X_r$ we get (see e.g. Lemma 3.20 in
\cite{bapat2010}) that
\begin{alignat}{4}\label{twentysevennew}
 \varrho(X^\top X)\leq \sum_{r=1}^c \varrho(X_r^\top X_r).
\end{alignat}
By the  properties of the tensor product we find that
\begin{alignat}{4}
X_r^\top X_r = X_{r, 1}^\top X_{r, 1} \otimes \ldots \otimes X_{r,
  d}^\top X_{r, d},
\end{alignat}
whose eigenvalues are of the form
$\alpha_{1,k_1}\alpha_{2,k_2}\cdots\alpha_{d,k_d}$, with
$\alpha_{j,k_j}$ being the $k_j$th  eigenvalue of  $X_{r,j}^\top
X_{r,j}$, see e.g. Theorem 4.2.12 in \cite{horn1991}. In particular,  
\begin{alignat*}{4}
\varrho(X_r^\top  X_r)=\prod_{j=1}^d\varrho(X_{r,j}^\top  X_{r,j}),  
\end{alignat*}
and this completes the proof. 
\end{proof}

Note that for $c = 1$ the upper bound is $\hat{L}^{(k)}
= \max(w_{i}^{(k)}) \prod_{j=1}^d\varrho(X_{1,j}^\top
X_{1,j})/n$, which is valid for any weight matrix. If the weight
matrix is itself a tensor product it is possible to compute the
Lipschitz constant exactly. Indeed, if $W^{(k)} = W^{(k)}_d \otimes
\ldots \otimes W^{(k)}_1$ then 
\begin{alignat*}{4}
X^\top W^{(k)} X = X_{1,d}^\top W^{(k)}_d X_{1,d} \otimes
\ldots \otimes X_{1,1}^\top W^{(k)}_1 X_{1,1},
\end{alignat*}
and by similar arguments as in the proof above,
\begin{alignat}{4}
L^{(k)} = \frac{1}{n} \prod_{j=1}^d\varrho(X_{1,j}^\top W^{(k)}_j X_{1,j}).
\end{alignat}

The outer loop is similar to the outer loop used in
e.g. the R packages \verb+glmnet+, \cite{friedman2010}, and
\verb+sglOptim+, \cite{vincent2014}. For completeness we demonstrate
that the outer loop with the stepsize determined by the Armijo rule 
is a special case of the algorithm treated in \cite{tseng2009}, which
implies a global convergence result of the outer loop. 

Following \cite{tseng2009} the Armijo rule gives the stepsize 
$\alpha_k \coloneqq b^j \alpha_0$, where $\alpha_{0} > 0$ and $b
\in (0,1)$ are given constants and  $j$ is determined as follows: With 
$d^{(k)} = \tilde{\theta}^{(k+1)} - \theta^{(k)}$ and
\begin{alignat}{4} \nonumber
\Delta_k \coloneqq - (u^{(k)})^\top X
d^{(k)} + \lambda (J(\tilde{\theta}^{(k+1)}) - J(\theta^{(k)})),
\end{alignat}  
then $j \in \mathbb{N}_0$ is the smallest nonnegative integer for which  
\begin{alignat}{4}
F(\theta^{(k)} + b^j \alpha_0 d^{(k)}) \leq
F(\theta^{(k)}) +  b^j \alpha_0 v \Delta_k,
\end{alignat}
where $ v \in (0,1)$ is a fixed constant.

\begin{thm} \label{thm:2}
  Let the stepsize, $\alpha_k$, be given by the Armijo rule
  above. If the design matrix $X$ has rank $p$ and if there exist
  constants $\bar c \geq \underline c > 0$ such that for all $k \in
  \mathbb{N}$ the diagonal weights in $W^{(k)}$, denoted
  $w_{i}^{(k)}$, satisfy 
\begin{alignat}{4}
\underline c \leq w_{i}^{(k)} \leq \bar c
\label{eq:wbound}
\end{alignat}
for $i =1, \ldots, n$, then $(F(\theta^{(k)}))$ is nonincreasing and any cluster point of $(\theta^{(k)})$ is a
stationary point of the objective function $F$. 
\end{thm}

\begin{proof} The theorem is a consequence of Theorem 1 (a) and
  (e) in \cite{tseng2009} once we have established that the search
  direction, $d^{(k)} = \tilde{\theta}^{(k+1)} - \theta^{(k)}$, coincides with the
  search direction defined by (6) in \cite{tseng2009}. 
Letting $d \coloneqq\theta-\theta^{(k)}$ denote a (potential) search direction we see that 
\begin{alignat*}{4}
\frac{1}{2n} \Vert &\sqrt{W^{(k)}}(X\theta-z^{(k)})\Vert^2_2\\
 &=\frac{1}{2n} (-(W^{(k)})^{-1}u^{(k)}+X(\theta-\theta^{(k)}))^\top W^{(k)}(-(W^{(k)})^{-1}u^{(k)}+X(\theta-\theta^{(k)})) \\
&=\frac{1}{2n} ((u^{(k)})^\top (W^{(k)})^{-1}u^{(k)}-(u^{(k)})^\top  Xd-d^\top X^\top u^{(k)}+d^\top X^\top W^{(k)} Xd)\\
& \propto 
- \underbrace{(u^{(k)})^\top X}_{\nabla_\theta l(\eta^{(k)})^\top} d+\frac{1}{2}d^\top \underbrace{X^\top
W^{(k)} X}_{H^{(k)}}d + C_k,
  \end{alignat*}
where $C_k$ is a constant not depending upon $\theta$. This shows that 
\begin{alignat}{4}
d^{(k)} = \argmin_{d \in \mathbb{R}^p} -\nabla_\theta
l(\eta^{(k)})^\top d+\frac{1}{2}d^\top H^{(k)} d + \lambda J(\theta^{(k)} + d),
\end{alignat}
and this is indeed the search direction defined by (6) in
\cite{tseng2009} (with the coordinate block consisting of all
coordinates). Observe that $H^{(k)} = X^T W^{(k)} X$ fulfills
Assumption 1 in \cite{tseng2009} by the assumptions that $X$ has rank
$p$ and that the weights are uniformly bounded away from 0 and
$\infty$. Therefore, all conditions for Theorem 1 in \cite{tseng2009}
are fulfilled, which completes the proof.
\end{proof} 

The convergence conclusion can be sharpened by making further
assumptions on the objective function and the weights.  

\begin{cor}  Suppose that  the weights are given by
\begin{alignat}{4}\label{nine}
  w_{i}^{(k)} = \vartheta'(\eta_i^{(k)})(g^{-1})'(\eta_i^{(k)}),\quad i = 1,\ldots,n.
\end{alignat}
If $X$ has rank $p$, if $F$ is level bounded,
  if the PMLE, $\theta^\ast$, is unique and if $(g^{-1})'$ is nonzero
 everywhere  it holds that $\theta^{(k)} \to \theta^\ast$ for $k \to \infty$. 
\end{cor}

\begin{proof} The sublevel set  $\Theta_0 \coloneqq \{\theta \mid F(\theta)\leq
  F(\theta^{(0)})\}$ is bounded by assumption, and it is closed because $J$ is closed and $-l$ is
  continuous. Hence, $\Theta_0$ is compact. Since the weights as a
  function of $\theta$,
\begin{alignat}{4}
\theta \mapsto
\vartheta'(\eta_i(\theta))(g^{-1})'(\eta_i(\theta))
\end{alignat}
for $i = 1, \ldots, n$, are continuous and strictly positive functions --
because $(g^{-1})'$ is assumed nonzero everywhere, see Appendix \ref{sec:expfam} -- they attain a strictly
positive minimum and a finite maximum over the compact set $\Theta_0$. This implies
that \eqref{eq:wbound} holds. Since $\theta^{(k)} \in \Theta_0$ and
$\theta^\ast$ is a unique stationary point in $\Theta_0$, it
follows from Theorem \ref{thm:2}, using again that $\Theta_0$ is
compact, that $\theta^{(k)} \to \theta^\ast$ for $k \to \infty$. 
\end{proof}

The weights given by \eqref{nine} are the common weights used for
GLMs, but exactly the same argument as above applies to other choices
as long as they are strictly positive and continuous functions of the
parameter $\theta$. A notable special case is  $w_i^{(k)} = 1$. Another
possibility, which is useful in the framework of GLAMs, is discussed
in Section \ref{sec:glamops}.

Observe that if $-l$ is strongly convex then $F$ is level bounded, $X$
has rank $p$ and $\theta^\ast$ is unique. If $X$ does not have rank
$p$, in particular, if $p > n$, we are not presenting any results on
the global convergence of the outer loop. Clearly, additional
assumptions on the penalty function $J$ must then be made to guarantee
convergence. 

\section{Implementation}\label{sec:glamops}

In this section we show how the computations required in the GD-PG
algorithm can be implemented to exploit the array
structure. The penalty function $J$ is not assumed to have any special
structure in general, and its evaluation is not discussed, but we do briefly discuss the computation of the
proximal operator for some special choices of $J$. We also describe
the R package, \verb+glamlasso+, which implements the
algorithm for  2 and 3-dimensional array models with the $\ell_1$-penalty  and the smoothly clipped absolute deviation (SCAD) penalty, and we present results
of further benchmark studies using simulated data. 

\subsection{Array operations}

The linear algebra operations needed in the GD-PG algorithm can all be
expressed in terms of two maps, $\Hop$ and $\Gop$, which are
defined below. The maps work directly on the tensor factors in terms
of $\rho$ defined in Appendix \ref{app:rho}. Introduce 
\begin{alignat}{4}\label{eighteen}
 \Hop( \langle X_{r, j} \rangle, \langle \Theta_r\rangle) \coloneqq \sum_{r=1}^c \rho(X_{r,d},\ldots,\rho(X_{r,1},\Theta_r)\ldots),
\end{alignat}
which gives an $n_1\times\cdots\times n_d$ array such that $\ve(\Hop(
\langle X_{r, j} \rangle, \langle \Theta_r\rangle))$ is the linear
predictor. Introduce also  
\begin{alignat}{4}\label{nineteen}
\Gop(\langle X_{r,j} \rangle , U) \coloneqq 
 \langle
 \rho(X_{1,d}^{\top},\ldots,\rho(X_{1,1}^{\top},U)\ldots),\ldots,\rho(X_{c,d}^{\top},\ldots,\rho(X_{c,1}^{\top},U)\ldots) \rangle
\end{alignat}
for $U$ an $n_1\times\cdots\times n_d$ array, which gives a tuple of
$c$ arrays. The map $G$ is used to carry out the gradient computation in \eqref{three}.

Below we describe how the linear algebra operations required in steps
2, 4 and 5 in Algorithm \ref{alg:cgdfpg} can be carried out using the two maps above. In doing
so we use ``$\equiv$'' to denote equality of vectors and arrays (or
tuples of arrays) up to a rearrangement of the entries. In the
implementation such a rearrangement is never required, but it gives a
connection between the array and vector representations of the
components in the algorithm.

\begin{description}
\item[Step 2:] The linear predictor is first computed,  
\begin{alignat}{4}
X^\top \theta^{(k)} \equiv \Hop(\langle X_{r,j} \rangle, \langle
\Theta_r^{(k)} \rangle).
\label{eq:linpre}
\end{alignat}
The array $V^{(k)}$ is computed by an entrywise computation,
e.g. by \eqref{nine}.  The arrays $U^{(k)}$ and $Z^{(k)}$ are computed by entrywise
 computations using \eqref{four} and \eqref{ten}, respectively. If the
 weights given by \eqref{nine} are used, $Z^{(k)}$ can be computed
 directly by \eqref{tennew} and $U^{(k)}$ does not need to be computed.
\item[Step 4:] In the inner PG loop the gradient,  $\nabla h$, must be
  recomputed in each iteration. To this end,
\begin{alignat}{4}
X^\top  W^{(k)}  z^{(k)}  & \equiv
\Gop(\langle X_{r,j} \rangle, V^{(k)}\odot Z^{(k)}) \label{twentyeight}
\end{alignat}
is precomputed. Here $\odot$ denotes the entrywise (Hadamard)
product. Then $\nabla h (\theta)$ is computed in terms of
\begin{alignat}{4}
X^\top W^{(k)}X\theta& \equiv \Gop(\langle X_{r,j} \rangle,
V^{(k)} \odot   \Hop(\langle X_{r,j} \rangle, \langle \Theta_{r} \rangle) ).\label{twentysixnew}  
\end{alignat}
\item[Step 5:] For the stepsize computation using the Armijo rule the
  linear predictor, 
\begin{alignat}{4}
X^\top \tilde{\theta}^{(k+1)} \equiv \Hop(\langle X_{r,j} \rangle, \langle
\tilde{\Theta}_r^{(k+1)} \rangle),
\label{eq:linpre2}
\end{alignat}
is first computed. The computation of $\Delta_k$ is achieved via
computing inner products of $U^{(k)}$ and the linear predictors
\eqref{eq:linpre} and \eqref{eq:linpre2}. The
line search then involves iterative recomputations of the linear
predictor via the map $\Hop$. 
\end{description}

If $\delta_k$ is not chosen sufficiently small to guarantee
convergence of the inner PG loop a line search must also be carried
out in step 4. To this end, repeated evaluations of $h$ are needed,
with $h(\theta)$ being computed as the weighted 2-norm of 
$ \Hop(\langle X_{r,j} \rangle, \langle \Theta_r \rangle) - Z^{(k)}$
with weights $V^{(k)}$. 

\subsection{Tensor product weights}

The bottleneck in the GD-PG algorithm is \eqref{twentysixnew}, which
is an expensive operation that has to be carried out repeatedly. If
the diagonal weight matrix is a tensor product, the computations
can be organized differently. This can reduce the run time,
especially when $p_{r,j} < n_j$. 

Suppose that $W^{(k)}=W_d^{(k)}\otimes\cdots\otimes W_1^{(k)}$, then 
\begin{alignat*}{4}
X_r^\top W^{(k)}X_m = X_{r,d}^\top W_d^{(k)}X_{m,d}\otimes\cdots\otimes
X_{r,1}^\top W_1^{(k)}X_{m,1}, \quad r,m = 1, \ldots,c.
\end{alignat*}
Hence $X^\top W^{(k)} X$ has tensor product blocks and
\eqref{twentysixnew} can be replaced by 
\begin{alignat}{4}
  X^{\top}W^{(k)}X\theta \equiv   
\langle \Hop(\langle X_{1,j}^\top W_j^{(k)}
X_{r,j} \rangle, \langle \Theta_r \rangle),\ldots,\Hop(\langle X_{c,j}^\top W_j^{(k)}
X_{r,j} \rangle, \langle \Theta_r \rangle) \rangle. 
\label{eq:kronweight}
\end{alignat}
The matrix products
$X_{r,k}^\top W_j^{(k)}X_{m,j}$ for $r,m = 1, \ldots, c$ and $j = 1,
\ldots, d$ can be precomputed in step 4. 
 
If the weight matrix is not a tensor product it might be approximated
by one so that \eqref{eq:kronweight} can be exploited. With $V^{(k)}$
denoting the weights in array form, then $V^{(k)}$ can be approximated by
$\hat{V}^{(k)}$, where 
\begin{alignat}{4}
\hat{V}^{(k)}_{i_1,\ldots, i_d} = \hat{v}^{(k)}_{1,i_1} \cdots
\hat{v}^{(k)}_{d,i_d},
\end{alignat}
with
\begin{alignat*}{4}
\hat{v}^{(k)}_{j, i_j} = \Bigg(\prod_{i_1,\ldots,i_{j-1},i_{j+1}, \ldots, i_d}
\frac{V^{(k)}_{i_1,\ldots,i_d}}{\overline{V}^{(k)}} \Bigg)^{\frac{1}{m_j}} =  \exp\bigg(
\frac{1}{m_j} \sum_{i_1,\ldots,i_{j-1},i_{j+1}, \ldots, i_d}  \log
V^{(k)}_{i_1,\ldots,i_d} - \log \overline{V}^{(k)} \bigg).
\end{alignat*}
Here $m_j = n/n_j = \prod_{j' \neq j} n_{j'}$ and 
\begin{alignat*}{4}
\overline{V}^{(k)} = \bigg( \prod_{i_1, \ldots, i_d} V_{i_1, \ldots,
  i_d} \bigg)^{\frac{1}{n}}.
\end{alignat*}
 The array $\hat{V}^{(k)}$ is equivalent to a diagonal weight matrix, which is a tensor product
of diagonal matrices with diagonals $(\hat{v}^{(k)}_{j,i})$. Observe
that if the weights in $V^{(k)}$ satisfy \eqref{eq:wbound} then so do
the approximating weights in $\hat{V}^{(k)}$.

\subsection{Proximal operations}

Efficient computation of the proximal
operator is necessary for the inner PG loop to be fast. Ideally $\mathrm{prox}_{\gamma}(z)$ should be
given in a closed form that is fast to evaluate. This is the case for
several commonly used penalty functions such as the 1-norm, the
squared 2-norm, their linear combination and several other separable
penalty functions.  

For the 1-norm, $\mathrm{prox}_{\gamma}(z)$ is given by soft
thresholding, see \cite{Beck2010} or \cite{Parikh2014}, that is,
\begin{alignat}{4}
\mathrm{prox}_{\gamma}(z)_i = (|z_i| - \gamma)_+ \mathrm{sign}(z_i).
\end{alignat}
For the squared 2-norm (ridge penalty) the proximal operator amounts
to multiplicative shrinkage,
\begin{alignat}{4}
\mathrm{prox}_{\gamma}(z)=\frac{1}{1+2\gamma} z,
\end{alignat}
see e.g. \cite{moreau1962}. For the elastic net penalty,
\begin{alignat}{4}
J(\theta) = ||\theta||_1 + \alpha||\theta||^2_2,
\end{alignat}
the proximal operator amounts to a composition of the proximal
operators for the 1-norm and the squared 2-norm, that is,
\begin{alignat}{4}
\mathrm{prox}_{\gamma}(z)_i = \frac{1}{1 + 2\alpha\gamma} (|z_i| - \gamma)_+ \mathrm{sign}(z_i),
\end{alignat}
see  \cite{Parikh2014}. For more examples see \cite{Parikh2014} and see also \cite{zhang2013}
for the proximal group shrinkage operator.

\subsection{The  {\upshape\texttt{glamlasso}}  R package}\label{sec:glamlasso}
The \verb+glamlasso+ R package provides an implementation of the GD-PG
algorithm for $\ell_1$-penalized  as well as SCAD-penalized
estimation in  2 and 3-dimensional GLAMs.  We note that as
  the SCAD penalty is non-convex the resulting optimization problem
  becomes non-convex and hence falls outside the original scope of our
  proposed method. However, by a local linear
  approximation to the SCAD penalty one obtains a weighted
  $\ell_1$-penalized problem. This is a convex problem, which may be
  solved within the framework proposed above. Especially, by
  iteratively solving a sequence of appropriately weighted
  $\ell_1$-penalized problems it is, in fact, possible to solve
  non-convex problems, see \cite{zou2008}. In the \verb+glamlasso+
  package this is implemented using the multistep adaptive lasso
  (MSA-lasso) algorithm from \cite{buhlmann2011}.  

The package is written in \verb+C+++ and utilizes the \verb+Rcpp+
package for the interface to \verb+R+, see
\cite{eddelbuettel2011}. At the time of writing this implementation
supports the Gaussian model with identity link, the
Binomial model with logit link, the Poisson model with log link and
the Gamma model with log link, but see \cite{lund2016} for the
current status.

The function \verb+glamlasso+ in the package solves the problem \eqref{five} with $J$ either given by the $\ell_1$-penalty or the SCAD penalty for a (user specified) number of penalty parameters $\lambda_{max}>\ldots>\lambda_{min}$. Here $\lambda _{max}$ is the infimum over the set of   penalty  parameters  yielding a zero  solution to \eqref{five} and $\lambda_{min}$ is a (user specified) fraction of $\lambda_{max}$.  For each model ($\lambda$-value) the algorithm is  warm-started by initiating the algorithm at the solution for the previous model. 

The interface of the function \verb+glamlasso+ resembles that of the \verb+glmnet+ function with some GD-PG  specific options. 

 The argument \verb+penalty+ controls the type of penalty to use. Currently the $\ell_1$-penalty (\verb+"lasso"+) and the SCAD penalty (\verb+"scad"+)  are implemented. 

 The argument \verb+steps+ controls the number of steps to use in the MSA algorithm when  the SCAD penalty is used.

The argument $\nu \in [0,1]$ (\verb+nu+)  controls the stepsize in the 
inner PG loop relative to the upper bound,  $\hat
L^{(k)}$, on the Lipschitz constant. Especially,  for $\nu\in (0,1)$ the stepsize is initially
$\delta^{(k)} \coloneqq 1/(\nu\hat L^{(k)})$ and the backtracking procedure
from \cite{beck2009} is employed only if divergence is detected.  For
$\nu = 1$ the stepsize is $\delta^{(k)} \coloneqq 1/\hat L_h$  and no
backtracking is done. For $\nu = 0$  the stepsize is initially
$\delta^{(k)} \coloneqq1$ and backtracking is done in each iteration. 

 The argument \verb+iwls = c("exact", "one", "kron1", "kron2" )+ specifies whether a
 tensor product approximation to the weights or the exact weights are
 used. The exact weights are the weights given by \eqref{nine}. Note
 that while a tensor product approximation may reduce the
 run time for the individual steps in the inner PG
 loop, it may also affect the
 convergence of the entire loop negatively.  

Finally, the argument \verb+Weights+ allows for a specification of
observation weights. This can be used -- as mentioned in
\cite{currie2006} -- as a way to model scattered (non-grid) data using
a GLAM by binning the data and then weighing each  bin according to
the number of observations in the bin. By setting some observation
weights to 0 it is also possible to model incompletely observed arrays
as illustrated in Section \ref{sec:incomplete}.

\subsection{Benchmarking on simulated data}\label{subsec:simexp}

To further investigate the performance of the GD-PG algorithm and its
implementation in \verb+glamlasso+ we carried out a benchmark study
based on simulated data  from a 3-dimensional GLAM. We report the setup and the results of the
benchmark study in this section.  See the  supplemental materials online for scripts used in this section.

For each $j\in \{1,2,3\}$  we  generated an  $n_j\times p_j$ matrix
$X_j$ by letting its rows be $n_j$ independent samples from a
$\mathcal{N}_{p_j}(0,\Sigma)$ distribution. The diagonal
entries of the covariance matrix $\Sigma$ were all equal to $\sigma
> 0$ and the
off diagonal elements were all equal to $\kappa$ for different choices
of $\kappa$. Since the design matrix $X = X_3 \otimes X_2 \otimes X_1$ is a tensor
product there is a non-zero correlation between the columns of $X$
even when $\kappa = 0$. Furthermore, each column of $X$ contains $n$
samples from a distribution with density given by a Meijer
$G$-function, see \cite{springer1970}.

We considered designs with $n_1=60r$, $n_2=20r$, $n_3=10r$ and
$p_1=\max\{3,n_1q\}$, $p_2=\max\{3,n_2q\}$, $p_3=\max\{3,n_3q\}$ for a
sequence of $r$-values and $q\in\{0.5, 3\}$. The number $q$ controls if
$p<n$ or $p>n$ and the size of the design matrix increases with $r$.  

The regression coefficients were generated as  
\begin{alignat*}{4}
\theta_m=(-1)^m\exp\Big(\frac{-(m-1)}{10}\Big)B_m,\quad m = 1,\ldots,p,
\end{alignat*}
where $B_1, \ldots, B_p$ are i.i.d. Bernoulli variables with $P(B_m =
1) = s$ for $s \in [0,1]$. Note that   $s$ controls the sparsity of the
coefficient vector and $s = 1$ results in a dense parameter vector. 

We generated observations from two different models for different
choices of parameters. 
\begin{description}
\item[Gaussian models:] We generated Gaussian observations with unit
  variance and the identity link with a dense parameter vector ($s =
  1$). The design was generated with $\sigma = 1$  and $\kappa \in \{0, 0.25\}$ for $p<n$ and $\kappa=0$ for $p>n $.
\item[Poisson models:] We generated Poisson observations with the
  log link function with a sparse parameter vector ($s = 0.01$). The
  design was generated with $\sigma = 0.71$ and  $\kappa \in \{0, 0.25\}$ for $p<n$  and $\kappa=0$ for $p>n $.  It is worth noting that  this quite artificial Poisson simulation setup  easily generates extremely large observations, which in turn can cause convergence problems for the algorithms, or even  NA values.
\end{description}

For each of the two models above and for the different combinations of
design and simulation parameters we computed the PMLE  
using \verb+glamlasso+ as well as \verb+glmnet+ for the same sequence of  $\lambda$-values.  The default length of this sequence is 100, however, both  \verb+glmnet+ and \verb+glamlasso+ will exit if  convergence is not obtained for some $\lambda$ value and return only the PMLEs for the preceding models along with the corresponding $\lambda $ sequence.

This benchmark study on simulated data was carried out on the same
computer as used for the benchmark study on real data as presented in
Section \ref{subsec:data}.  However, here we ran the simulation
  and optimization procedures five times for each size and parameter
  combination and report the run times along with their means as well
  as the mean relative deviations of the objective functions. See Section \ref{subsec:data} for other
  details on how \verb+glamlasso+ and \verb+glmnet+ were compared and
  Figures \ref{fig:one}, \ref{fig:two} and \ref{fig:three} below present the
  results.

\begin{figure}
\begin{center}
\includegraphics[scale = 0.45]{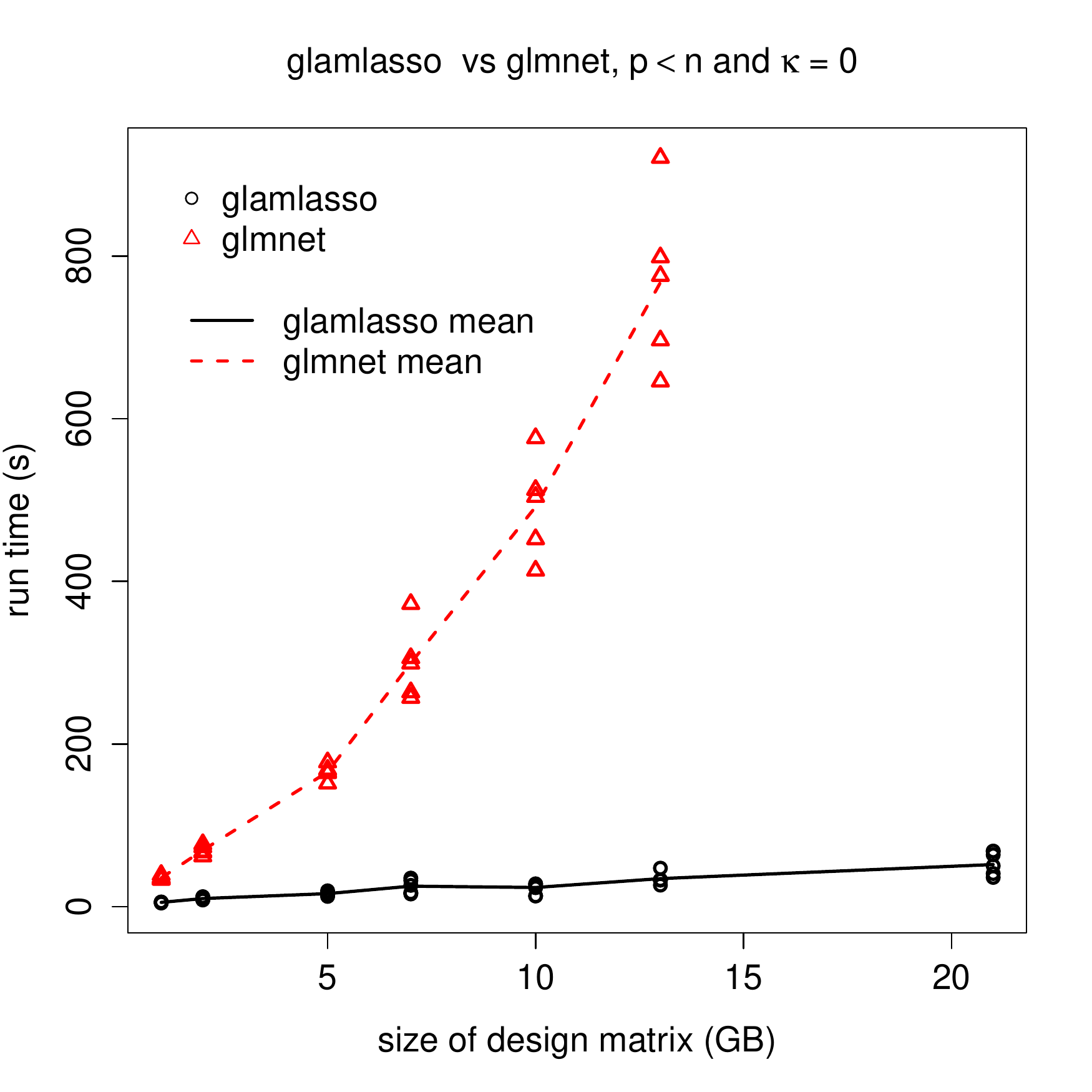}
\includegraphics[scale = 0.45]{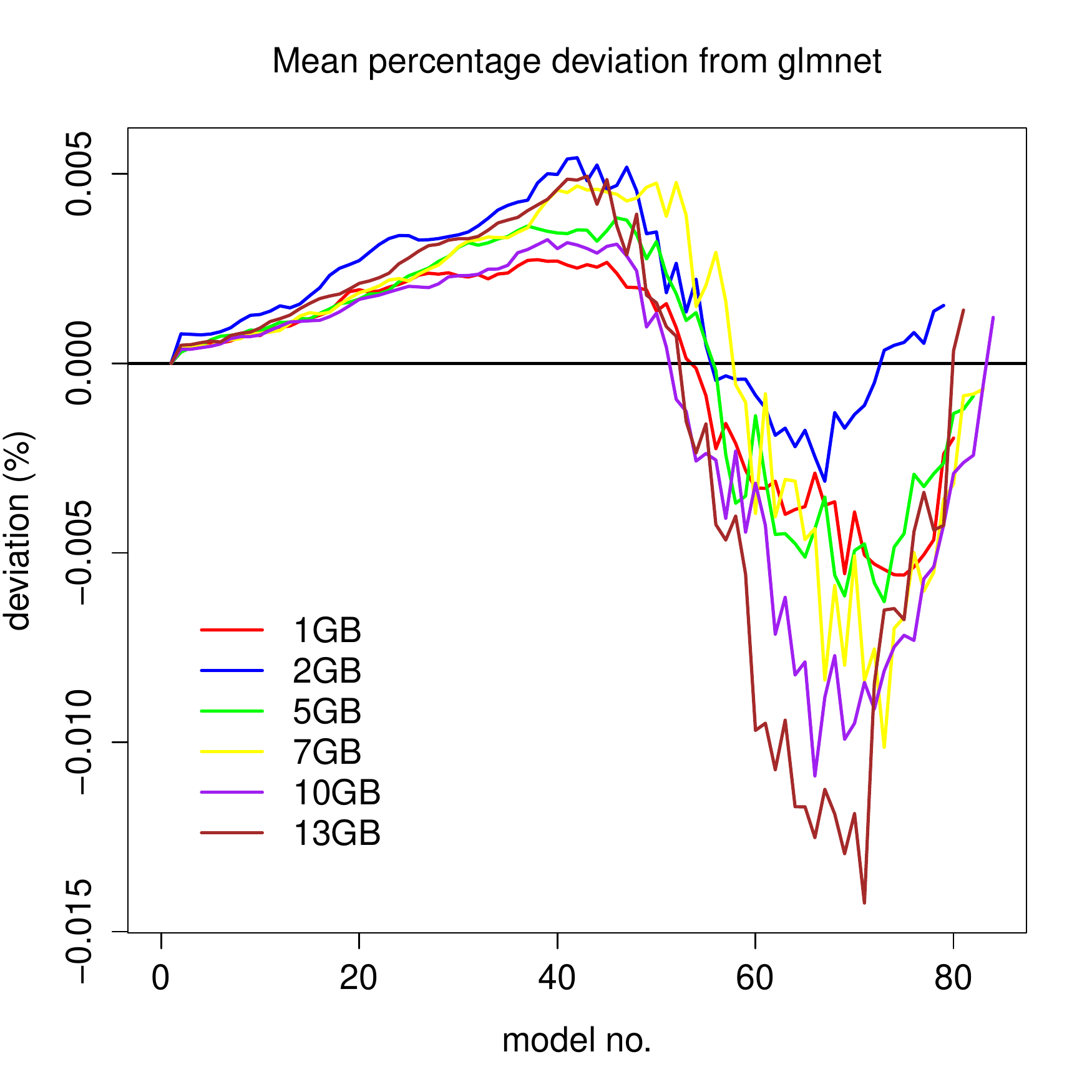}\\
\includegraphics[scale = 0.45]{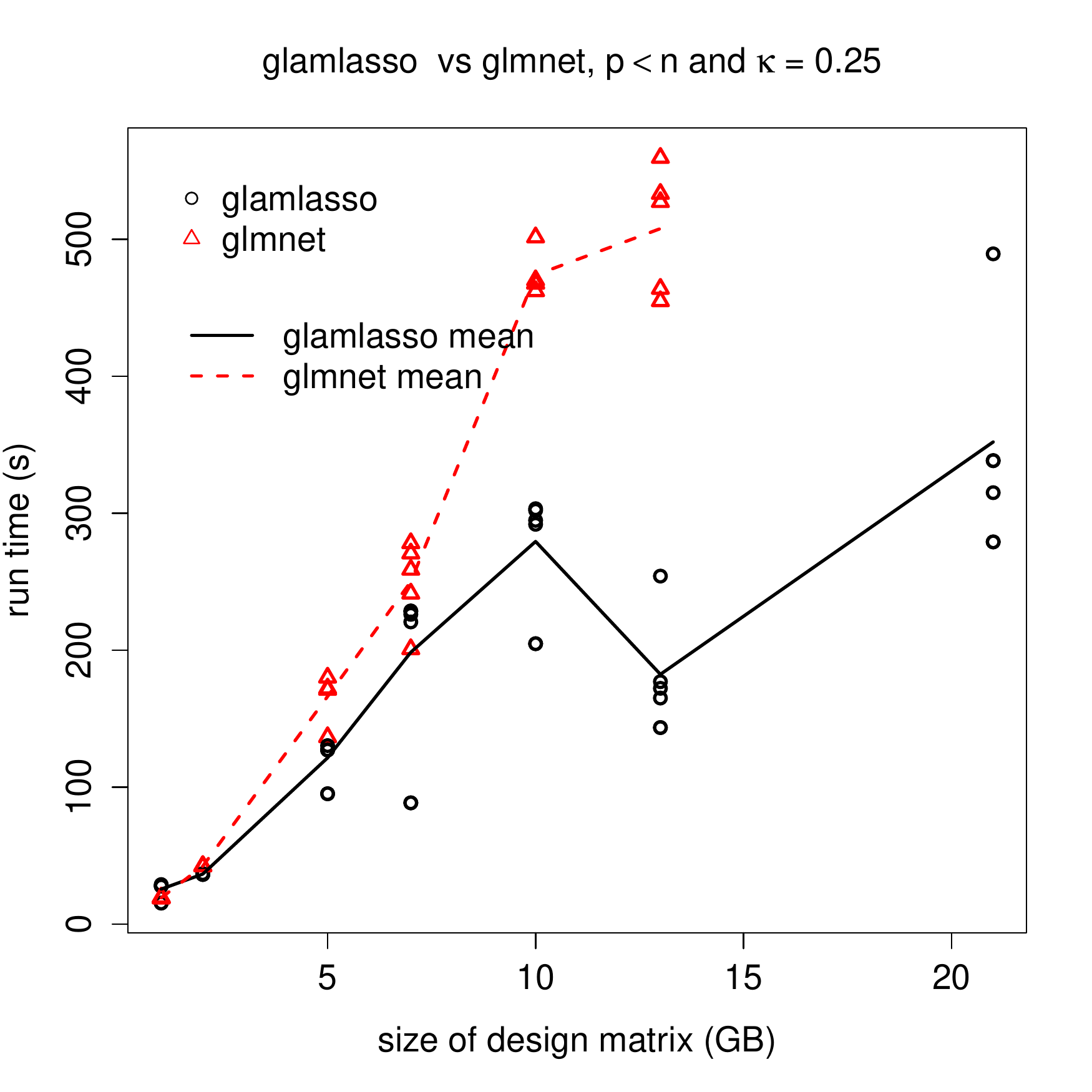}
\includegraphics[scale = 0.45]{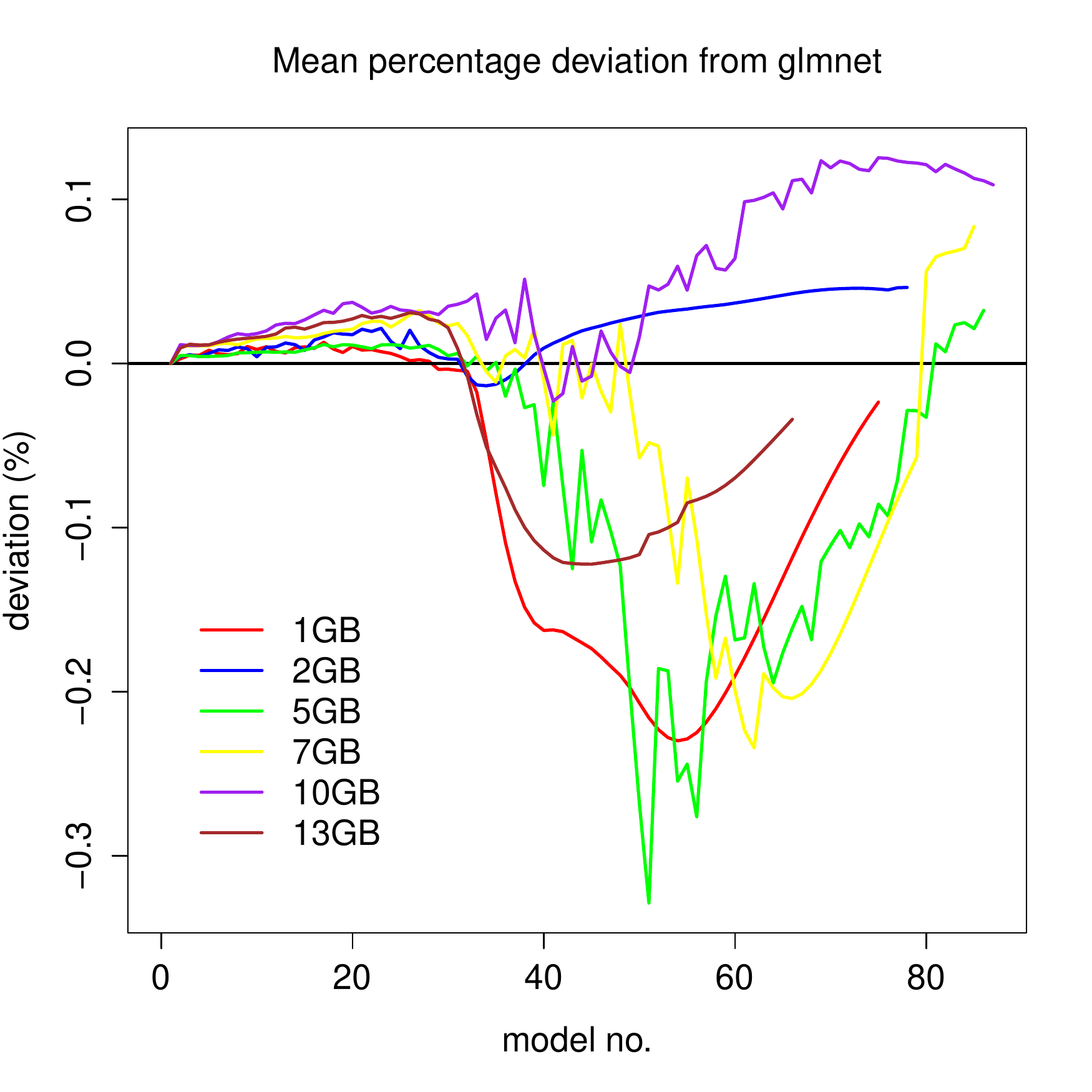}
\caption{Benchmark results for the Gaussian models and $p < n$. Run time in
  seconds is shown as a function of the size of the design matrix in
  GB (left). Relative mean deviation in the attained objective function
  values as given by (20) is shown as a function of model number
  (right).  The top row gives the results for $\kappa=0$ and the
  bottom for $\kappa=0.25$.}
\label{fig:one}
\end{center}
\end{figure}

\begin{figure}
\begin{center}
\includegraphics[scale = 0.45]{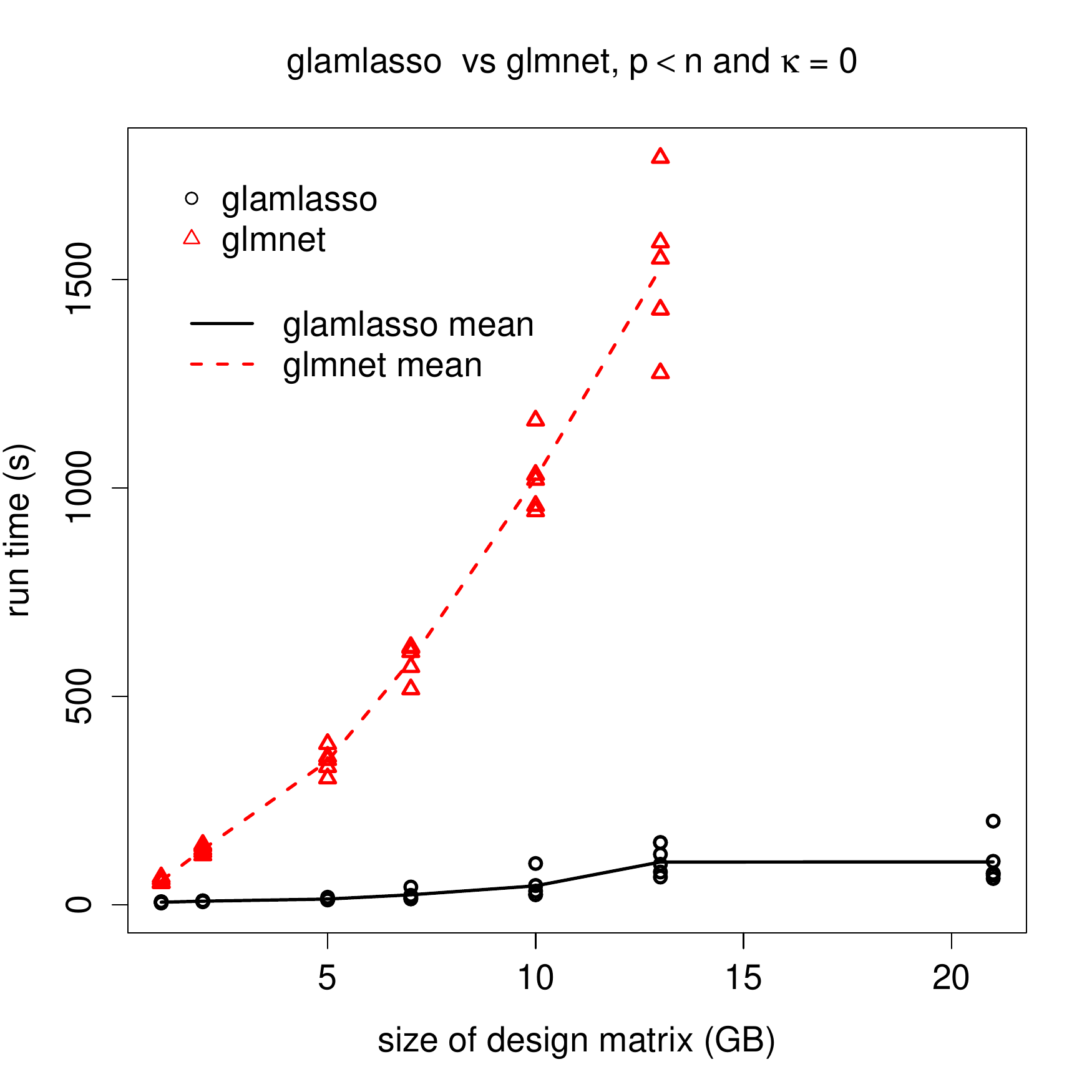}
\includegraphics[scale = 0.45]{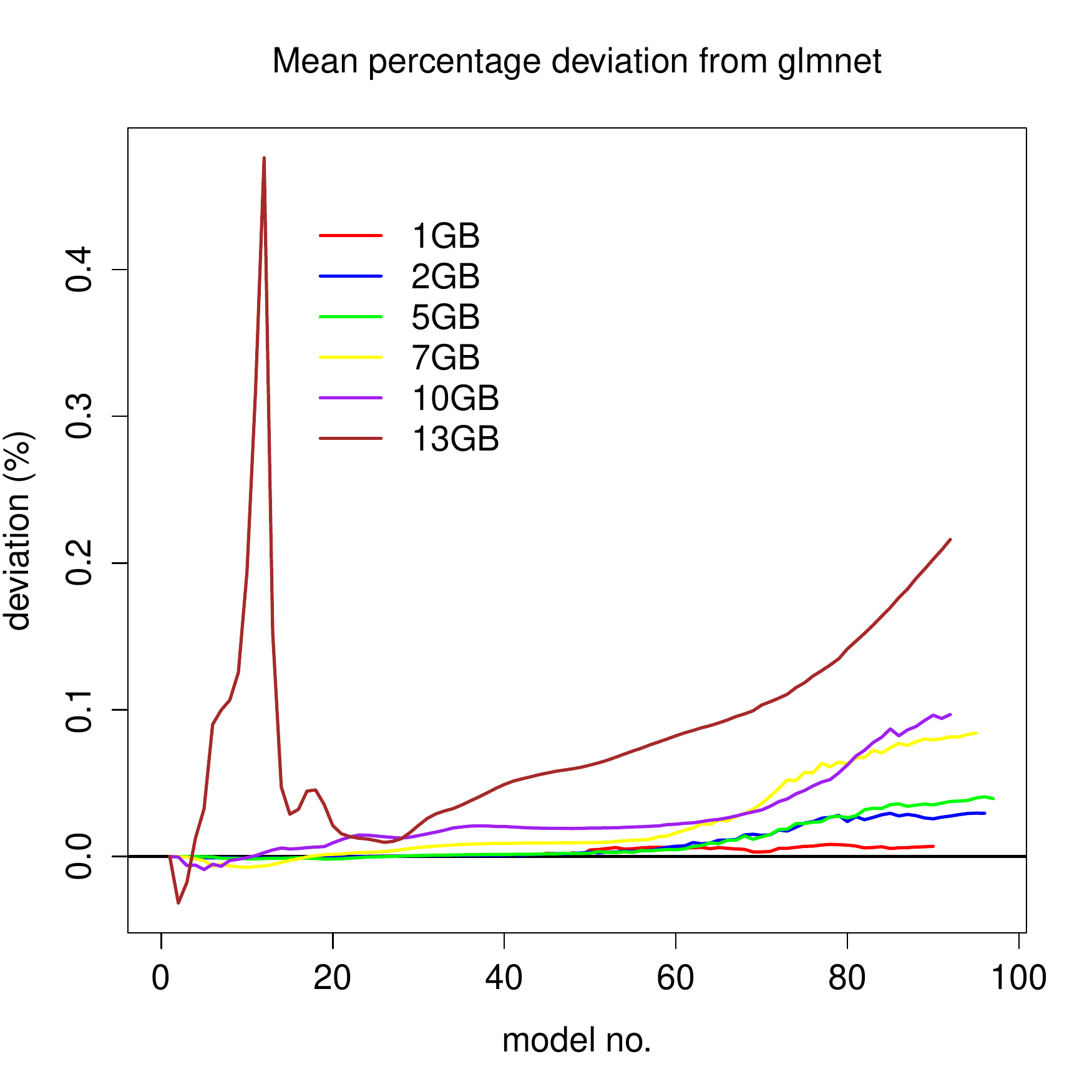}
\includegraphics[scale = 0.45]{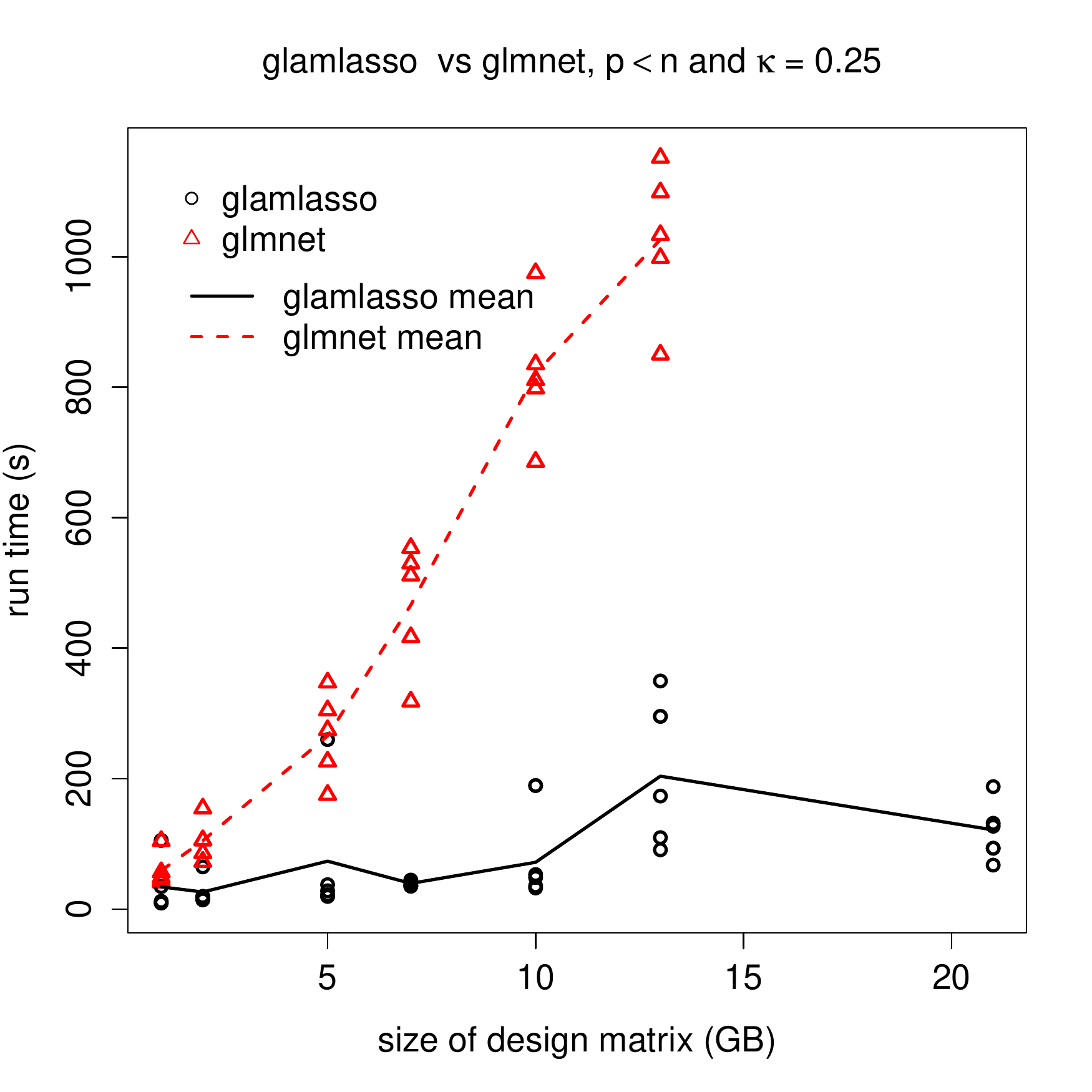}
\includegraphics[scale = 0.45]{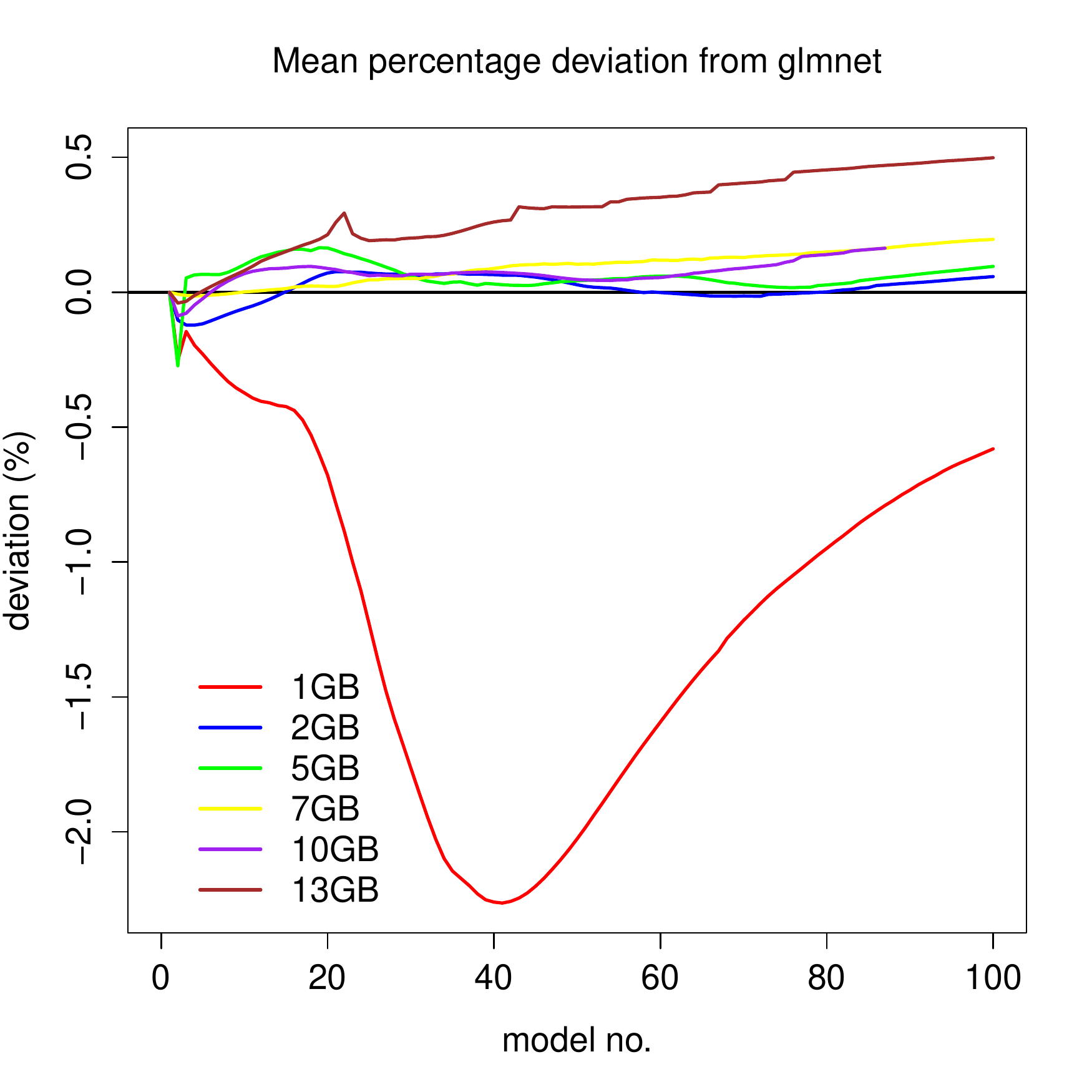}
\caption{Benchmark results for the Poisson models and $p < n$. Run time in
  seconds is shown as a function of the size of the design matrix in
  GB (left). Relative mean deviation in the attained objective function
  values as given by (20) is shown as a function of model number
  (right). The top row gives the results for $\kappa=0$ and the
  bottom for $\kappa=0.25$.}
\label{fig:two}
\end{center}
\end{figure}

\begin{figure}
\begin{center}
\includegraphics[scale = 0.45]{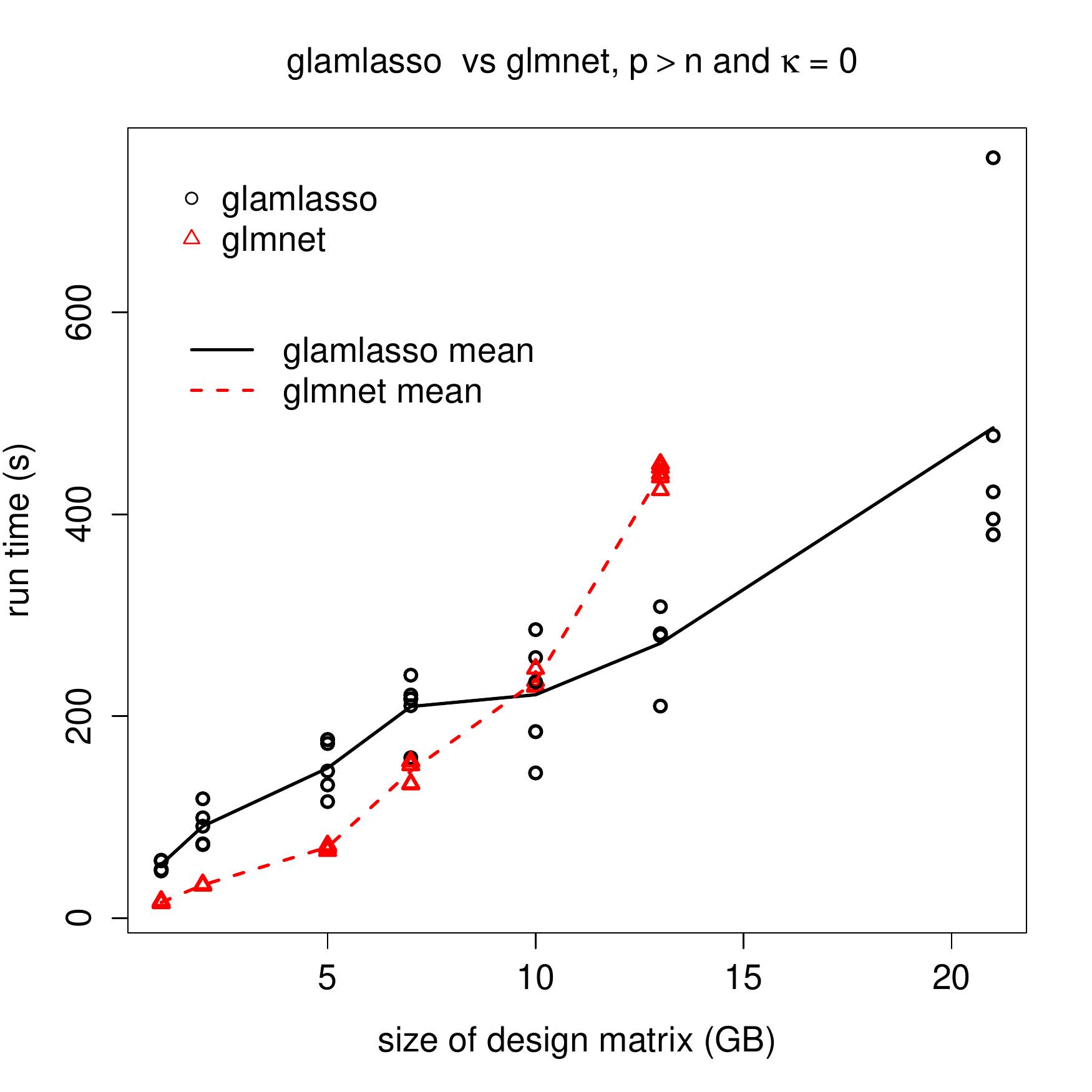}
\includegraphics[scale = 0.45]{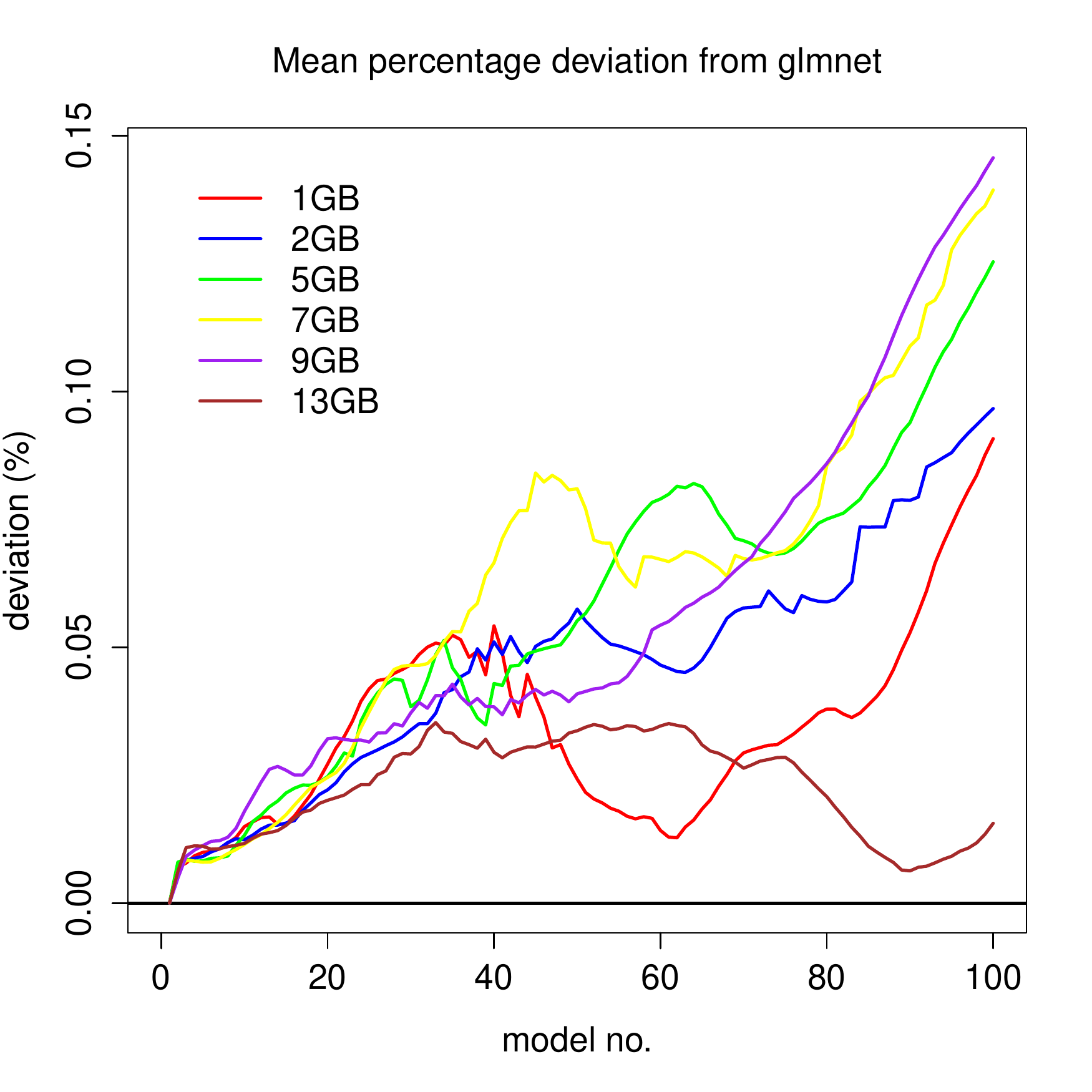}\\
\includegraphics[scale = 0.45]{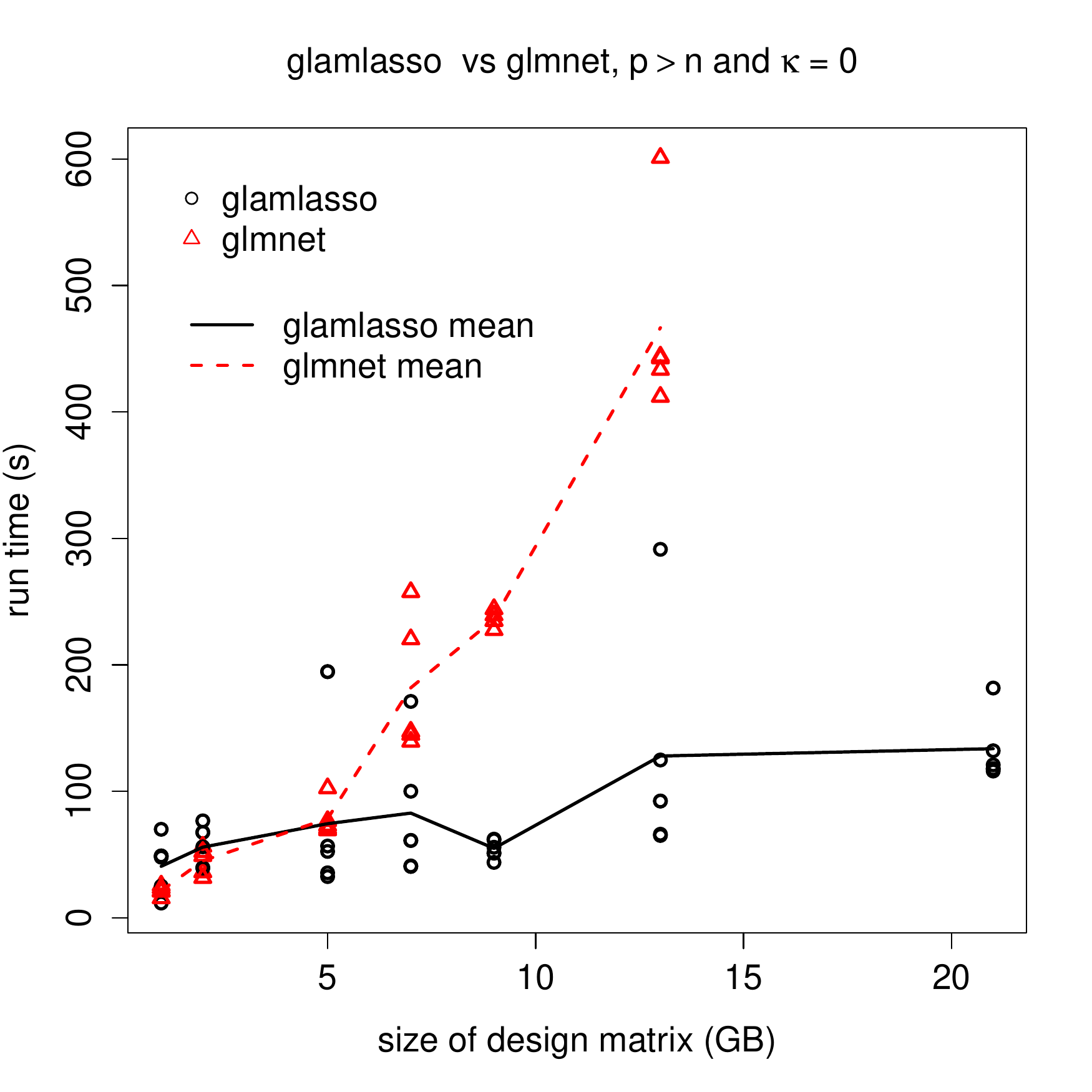}
\includegraphics[scale = 0.45]{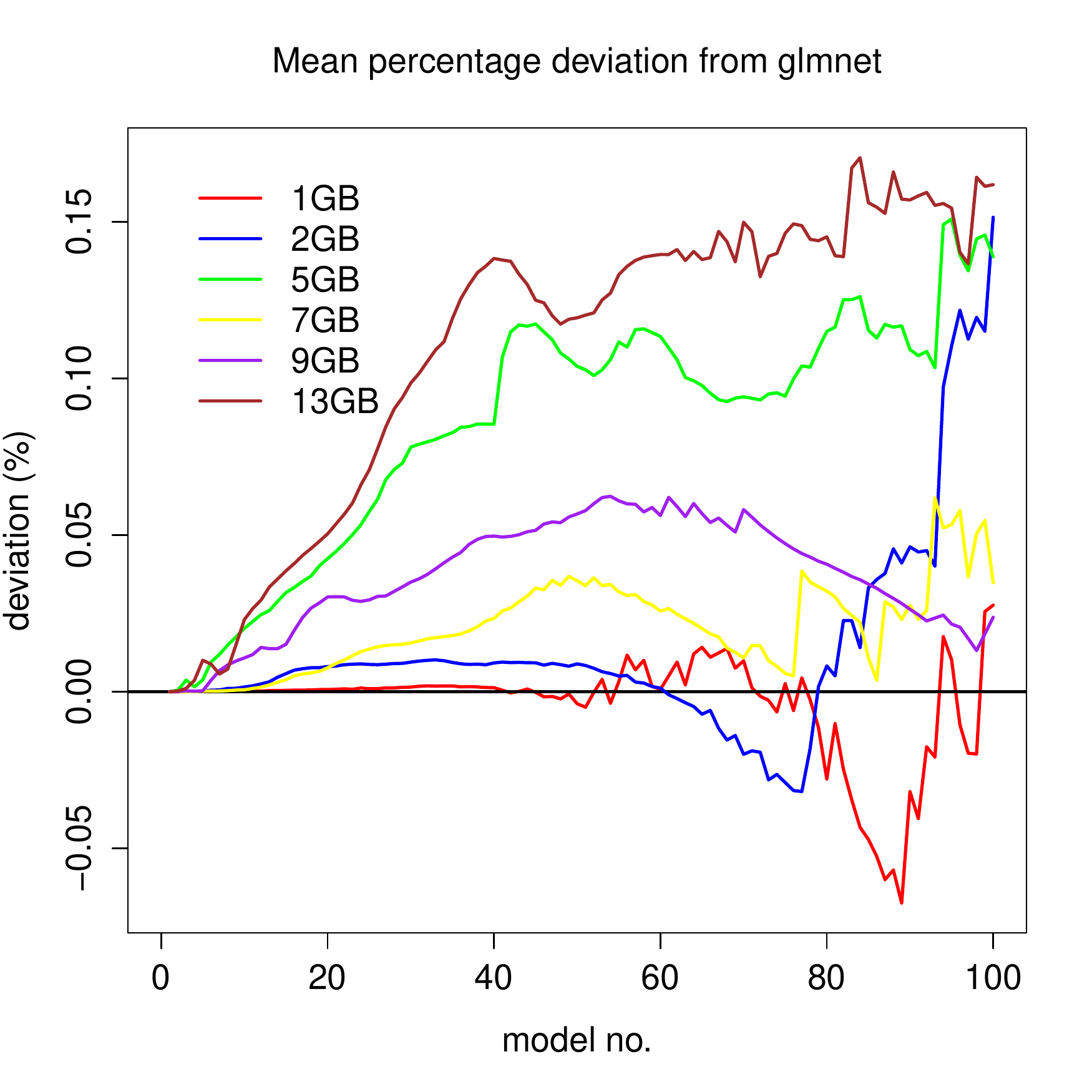}
\caption{Benchmark results for $p > n$. Run time in
  seconds is shown as a function of the size of the design matrix in
  GB (left). Relative mean deviation in the attained objective function
  values as given by (20) is shown as a function of model number
  (right).  The top row gives the results for the Gaussian model and the
  bottom for Poisson model.}
\label{fig:three}
\end{center}
\end{figure}

Figure \ref{fig:one} shows the results for the Gaussian models for $p
< n$. Here \verb+glamlasso+ generally outperformed
\verb+glmnet+ in terms of run time -- especially for $\kappa = 0$.  It scaled well with the size of the design
matrix and it could fit the model
for large design matrices that \verb+glmnet+ could not handle. 

It should be noted that for the Gaussian models with the identity link
there is no outer loop, hence the comparison is in this case
effectively between the  (GLAM enhanced) proximal gradient
algorithm and the coordinate descent algorithm as implemented in
\verb+glmnet+.
 
Figure \ref{fig:two} shows the results for the Poisson models for $p <
n$. As for the Guassian case, \verb+glamlasso+ was generally faster
than \verb+glmnet+. The run times for
\verb+glamlasso+ also scaled  very well with the size of the
design matrix for both values of $\kappa$. 

Figure \ref{fig:three} shows the results for both models for $p > n$ and
$\kappa = 0$. Here the run times were comparable for small design
matrices, with \verb+glmnet+ being a little faster for the Gaussian model,
but \verb+glamlasso+ stilled scaled better with the size of the design
matrix. For $\kappa > 0$ (results not shown) \verb+glamlasso+ retained
its benefit in terms of memory usage, but \verb+glmnet+ became
comparable or even faster for the Gaussian model than \verb+glamlasso+.

 In the comparisons above we have not included the time it took to construct the actual
  design matrix for the \verb+glmnet+ procedure. However, the
  construction and handling of matrices, whose size is a substantial
  fraction of the computers memory, was quite time consuming (between 15 minutes and
  up to one hour) underlining the advantage of our design matrix free
  method.

\section{Discussion}\label{sec:disc}

The algorithm implemented in the R package \verb+glmnet+ and described in
\cite{friedman2010} computes the penalized and weighted
least squares estimate given by \eqref{eleven} by a coordinate descent
algorithm. For penalty functions like the 1-norm that induce sparsity
of the minimizer, this is recognized as a very efficient
algorithm. Our initial strategy was to adapt the coordinate descent
algorithm to GLAMs so that it could take advantage of the tensor
product structure of the design matrix. It turned out to be
difficult to do that. It is straight forward to implement a memory
efficient version of the coordinate descent algorithm that does
not require the storage of the full tensor product design matrix, but
it is not obvious how to exploit the array structure to reduce the
computational complexity. Consequently, our implementation of such an
algorithm was outperformed by \verb+glmnet+ in terms of run time, and
for this reason alternatives to the coordinate descent algorithm were
explored. 

Proximal gradient algorithms for solving nonsmooth optimization
problems have recently received renewed attention. One reason is that
they have shown to be useful for large-scale data analysis problems, 
see e.g. \cite{Parikh2014}. In the image analysis literature  the
proximal gradient algorithm for a squared error loss with an
$\ell_1$-penalty is known as ISTA
(iterative selection-thresholding algorithm), see \cite{beck2009} and
\cite{Beck2010}. The accelerated version with a specific acceleration sequence 
was dubbed FISTA (fast ISTA) by \cite{beck2009}. For small-scale problems and unstructured design
matrices it is our experience that the coordinate descent algorithm outperforms 
accelerated proximal algorithms like FISTA.  This observation is
  also in line with the more  systematic comparisons presented in
  Section 5.5 in \cite{hastie2015}. For large-scale problems
and/or structured design matrices -- such as the tensor product design
matrices considered in this paper -- the proximal gradient algorithms
may take advantage of the structure. The Gaussian smoothing
example demonstrated that this is indeed the
case. 

When the squared error loss is replaced by the negative log-likelihood
our proposal is similar to the approach taken in \verb+glmnet+, where 
penalized weighted least squares problems are solved iteratively by an
inner loop. The main difference is that we suggest using a proximal
gradient algorithm instead of a coordinate descent algorithm for the
inner loop. Including weights is  only a trivial modification of FISTA from
\cite{beck2009}, but the weight matrix commonly used for fitting GLMs is
not a tensor product. Despite of this it is still possible to exploit the tensor
product structure to speed up the inner loop, but by making a tensor
approximation to the weights we obtained in some cases further improvements. For
this reason we developed the GD-PG algorithm with an arbitrary choice of weights. The Poisson
smoothing example demonstrated that when compared to coordinate
descent the inner PG loop was capable of taking advantage of the tensor
product structure. 

The convergence analysis combines general results from the
optimization literature to obtain convergence results for the inner proximal algorithm and the
outer gradient based descent algorithm. These results are strongest
when the design matrix has rank $p$ (thus requiring $p \leq
n$). Convergence for $p > n$ would require additional
assumptions on $J$, which we have not explored in any detail. Our
experience for $J = \Vert \cdot \Vert_1$ is that the algorithm converges in
practice also when $p > n$. Our most important contribution to the
convergence analysis is the computation of the upper bound
$\hat{L}^{(k)}$ of the Lipschitz constant $L^{(k)}$. This upper bound
relies on the tensor product structure. For large-scale problems the
computation of $ L^{(k)}$ will in general be infeasible due to the
size of  $X^{\top} W^{(k)}X$. However, for
the tensor product design matrices considered, the upper bound is
computable, and a permissible stepsize
$\delta^{(k)}$ that ensures convergence of the inner PG loop can be
chosen. 

It should be noted that the GD-PG algorithm requires minimal
assumptions on $J$, but that the proximal operator associated with $J$
should be fast to compute for the algorithm to be efficient. Though it
has not been explored in this paper, the generality allows for the
incorporation of convex parameter contraints. For box constraints $J$
will be separable and the proximal operator will be fast to compute. 

The simulation study confirmed what the smoothing applications had
showed, namely that the GD-PG algorithm with $J = \Vert \cdot \Vert_1$
and its implementation in the R package \verb+glamlasso+ scales well
with the problem size. It can, in particular, efficiently handle
problems where the design matrix becomes prohibitively large to be
computed and stored explicitly. Moreover, in the simulation study the
run times were in most cases smaller than or comparable to that of
$\verb+glmnet+$ even for small problem sizes.  However, the
  simulation study also revealed that when $p > n$ the run time
  benefits of $\verb+glamlasso+$ over $\verb+glmnet+$ were small or
  dimished completely -- in
  particular for small problem sizes. One explanation could be that
  $\verb+glmnet+$ implements a screening rule, which is particularly
  beneficial when $p > n$. It appears to be difficult to combine such
  screening rules with the tensor product structure of the design
  matrix. When $p < n$, as in the smoothing applications,
\verb+glamlasso+ was, however, faster than \verb+glmnet+ and scaled
much better with the size of the problem. This was true even when a
sparse representation of the design matrix was used, though
\verb+glmnet+ was faster and scaled better with the size of the design
matrix in this case for both examples. It should be noted that
\verb+glamlasso+ achieves its performance without relying on sparsity
of the design matrix, and it thus works equally well for smoothing
with non-local as well as local basis functions.

In conclusion, we have developed and implemented an algorithm for
computing the penalized maximum likelihood estimate for a
GLAM. When compared to \cite{currie2006} our focus has been on nonsmooth penalty functions that
yield sparse estimates. It was shown how the proposed GD-PG algorithm
can take advantage of the GLAM data structure, and it was demonstrated
that our implementation is both time and memory efficient. The smoothing examples illustrated how
GLAMs can easily be fitted to 3D data on a standard laptop computer
using the R package \verb+glamlasso+.    

 \section{Supplementary Materials}\label{sec:supp}
\begin{description}
\item[SuppMatJCGS] SuppMatJCGS is a folder containing scripts and datasets used in the  examples in sections \ref{subsubsec:neuro}, \ref{subsubsec:taxi},  \ref{sec:incomplete} and \ref{subsec:simexp} along with a ReadMe file. (SuppMatJCGS.zip, zipped file).
\end{description}

\bibliographystyle{chicago} 
\bibliography{Bibliotek}
 
\appendix

\newpage

\section{The maps $\ve$ and $\rho$} \label{app:rho}

The map $\ve$ maps an $n_1 \times \ldots \times n_d$ array to a $\prod_{i=1}^d n_d$-dimensional vector. This is sometimes known as
``flattening'' the array. For $j = 1, \ldots, d$ and $i_j
= 1, \ldots, n_j$ introduce the integer
\begin{alignat}{4}
[i_1, \ldots, i_d] \coloneqq i_1 + n_1((i_2 - 1) + n_2((i_3 - 1) +
\ldots n_{d-1} (i_d - 1) \ldots )).
\label{index}
\end{alignat} 
Then $\ve$ is defined as 
\begin{alignat}{4}
\ve(A)_{[i_1, \ldots, i_d]} \coloneqq A_{i_1, \ldots, i_d}
\end{alignat}
for an array $A$. This definition of $\ve$ corresponds to flattening
a matrix in column-major order.

Following the definitions in \cite{currie2006} (see also \cite{deboor1979} and
\cite{buis1996}), $\rho$ maps an $r
\times n_1$ matrix and an $n_1 \times \ldots \times n_d$ array to
an $n_2 \times \ldots \times n_d \times r$ array. With $X$ the
matrix and $A$ the array then 
\begin{alignat}{4}
\rho(X, A)_{i_1, \ldots, i_d} \coloneqq \sum_{j} X_{i_d, j} A_{j, i_1, \ldots, i_{d-1}}.
\end{alignat}
From this definition it follows directly that 
\begin{alignat*}{4}
(X_d \otimes \ldots \otimes X_1) \ve(A)_{[i_1,\ldots, i_d]}  &= \sum_{j_1,
  \ldots, j_d} X_{d,i_d,j_d} \cdots  X_{1,i_1,j_1} A_{j_1,\ldots,
  j_d}  \\
&=  \sum_{j_d} X_{d,i_d,j_d} \cdots \sum_{j_2} X_{2, i_2, j_2}
\sum_{j_1}  X_{1,i_1,j_1} A_{j_1,\ldots, j_d}  \\
&=  \rho(X_d, \ldots, \rho(X_2, \rho(X_1, A))\ldots)_{i_1,\ldots, i_d} 
\end{alignat*}
where $[i_1,\ldots, i_d]$ denotes the index defined by \eqref{index}.

\section{Exponential families} \label{sec:expfam}
The exponential families considered are distributions on
$\mathbb{R}$ whose density is
\begin{alignat*}{4}
  f_{\vartheta,\psi}(y)=\exp\Big(\frac{a(\vartheta y-b(\vartheta))}{\psi}\Big)
\end{alignat*}
w.r.t. some reference measure. Here $\vartheta$ is the  canonical
(real valued) parameter, $\psi > 0$ is the
dispersion parameter, $a > 0$ is a known and fixed weight and $b$ is
the log-normalization constant as a function of $\vartheta$ that ensures that the
density integrates to 1. In general,  $\vartheta$ may have to be
restricted to an interval depending on the reference measure
used. Note that the reference measure will depend upon
$\psi$ but not on $\vartheta$. 

With $\eta$ denoting the linear predictor in a generalized linear
model we regard $\vartheta(\eta)$ as a parameter function that maps
the linear predictor to the canonical parameter, such that the mean
equals $g^{-1}(\eta)$ when $g$ is the link function. From this it
can easily be derived that $b'(\vartheta(\eta)) = g^{-1}(\eta)$. For a canonical
link function, $\vartheta(\eta) = \eta$ and $b' = g^{-1}$. In terms of $\eta$ the
log-density can be written as 
\begin{alignat*}{4}
  \log f_{\vartheta(\eta),\psi}(y) \propto a (\vartheta(\eta)y-b(\vartheta(\eta))).
\end{alignat*}
From this it follows that 
\begin{alignat}{4}
\partial_{\eta} \log f_{\vartheta(\eta),\psi}(y)  =a \vartheta'(\eta)( y - g^{-1}(\eta)),
\end{alignat}
and the score statistic, $u = \nabla_{\eta} l(\eta)$, entering in \eqref{three}
is thus given by 
\begin{alignat}{4} \label{four}
u_i = a_i \vartheta'(\eta_i)( y_i - g^{-1}(\eta_i)), \quad i =
1, \ldots n. 
\end{alignat}

The weights commonly used when fitting a GLM are 
\begin{alignat}{4}\label{weights}
w_i = \vartheta'(\eta_i) (g^{-1})'(\eta_i),
\end{alignat}
which are known to be strictly positive provided that $(g^{-1})'$ is
nonzero everywhere (thus $g^{-1}$ is strictly monotone). This is not
entirely obvious, but $w_i$ is the variance of $u_i$ (with $a_i =
1$ and $\psi = 1$), which is nonzero whenever $(g^{-1})'$ is nonzero everywhere. 

We may note that when the weights are given by \eqref{weights}, the working
response $z$, see \eqref{ten}, given the linear predictor $\eta$ can be computed as
\begin{alignat}{4}\label{tennew}
z_i = a_i(y_i- g^{-1}(\eta_i))g'(g^{-1}(\eta_i))+\eta_i,
\end{alignat}
which renders it unnecessary to compute the intermediate score statistic.


%
%
%
%
%
%
%
%
%
\end{document}